\documentclass[11pt]{amsart}
\usepackage{amsmath,amsxtra,amssymb,latexsym, amscd,amsthm}
\usepackage[colorlinks]{hyperref}
\usepackage{tikz}
\usepackage{xfrac}
\usepackage{nicefrac}
\usetikzlibrary{arrows}
\usepackage{bbm}
\usepackage{gensymb}
\usepackage{graphicx}
\usepackage{float}
\numberwithin{equation}{section}
\numberwithin{figure}{section}
\newtheorem{theorem}{Theorem}[section]
\newtheorem{lemma}[theorem]{Lemma}

\newtheorem{remark}[theorem]{Remark}

\begin{document}

\newcommand{\cD}{\mathcal D}
\newcommand{\cR}{\mathcal R}
\newcommand{\cC}{\mathcal C}
\newcommand{\cB}{\mathcal B}
\newcommand{\cA}{\mathcal A}
\newcommand{\cM}{\mathcal M}
\newcommand{\cN}{\mathcal N}

\newcommand{\nn}{\mathbb N}
\newcommand{\cc}{\mathbb{C}}
\newcommand{\zz}{\mathbb Z}
\newcommand{\rr}{\mathbb R}
\newcommand{\ringz}{\mathring{\Sigma_{\mathbb Z}}}
\newcommand{\ve}{\varepsilon}
\newcommand{\tit}{\tilde{t}}
\newcommand{\tif}{\tilde{f}}

\newcommand{\tiD}{\widetilde{\mathcal D}}

\newcommand{\ld}{\lambda}
\newcommand{\til}{\tilde{\lambda}}
\newcommand{\tib}{\tilde{b}}
\newcommand{\tia}{\tilde{a}}
\newcommand{\one}{\mathbbm{1}_{I}(M_{\Lambda})}
\title{Resonances for 1D half-line periodic operators: II. Special case}
\author{Trinh Tuan Phong}
\begin{abstract}
The present paper is devoted to the study of resonances for a  $1$D Schr\"{o}dinger operator with truncated periodic potential.
Precisely, we consider the half-line operator $H^{\nn}=-\Delta +V$ and $H^{\nn}_{L}= -\Delta + V\mathbbm{1}_{[0, L]}$
acting on $\ell^{2}(\mathbb N)$ with Dirichlet boundary condition at $0$ with $L \in \nn$. We describe the resonances of $H^{\nn}_{L}$ near the boundary of the essential spectrum of $H^{\nn}$ as $L \rightarrow +\infty$ under a special assumption. 
\end{abstract}
\maketitle

\section{Introduction}\label{S:Intro}
Let $V$ be a periodic potential of period $p$ and $-\Delta$ be the (negative) discrete Laplacian on $l^{2}(\mathbb Z)$. 
 We define the $1$D Schr\"{o}dinger operators $H^{\mathbb Z}:=-\Delta +V$ acting on $l^{2}(\mathbb Z)$:
 \begin{equation}\label{eq:9.1}
(H^{\mathbb Z}u)(n)= \left ((-\Delta +V)u\right) (n)= u(n-1)+u(n+1) +V(n)u(n), \;\; \forall n\in \mathbb Z
\end{equation} 
 and $H^{\mathbb N}:=-\Delta+V$ acting on $l^{2}(\mathbb N)$ with Dirichlet boundary condition (b.c.) at $0$.\\
 Denote by $\Sigma_{\mathbb Z}$ the spectrum of $H^{\mathbb Z}$ and $\Sigma_{\mathbb N}$ the spectrum of $H^{\mathbb N}$. One has the following description for the spectra of $H^{\bullet}$ where $\bullet \in \{\mathbb N, \mathbb Z\}$:  
\begin{itemize}
\item[$\bullet$]  $\Sigma_{\mathbb Z}$ is a union of disjoint intervals; the spectrum of $H^{\mathbb Z}$ is purely absolutely continuous (a.c.) and the spectral resolution can be obtained via the Bloch-Floquet decomposition (see \cite{pierre76} for more details).
\item[$\bullet$] $\Sigma_{\mathbb N}=\Sigma_{\mathbb Z} \cup \{v_i\}_{i=1}^m$ where $\Sigma_{\mathbb Z}$ is the a.c. spectrum of $H^{\mathbb N}$ and $\{v_i\}_{i=1}^m$ are isolated simple eigenvalues of $H^{\nn}$ associated to exponentially decaying eigenfunctions (c.f. \cite{pavlov94}).
\end{itemize}
Pick a large natural number $L$, we set:
$$ \text{$H^{\mathbb N}_L:=-\Delta+ V\mathbbm{1}_{[0, L]}$ on $l^2(\mathbb N)$ with Dirichlet boundary condition (b.c.) at $0$. }$$
It is easy to check that the operator $H^{\nn}_L$ is self-adjoint. Then, the resolvent $z\in \mathbb C^{+}\mapsto (z-H^{\nn}_L)^{-1}$ is well defined  on $\mathbb C^{+}$. Moreover, one can show that $z \mapsto (z-H^{\nn}_L)^{-1}$ admits a meromorphic continuation from $\mathbb C^{+}$ to $\mathbb C \backslash \left((-\infty, -2] \cup [2, +\infty)\right) $ with values in the self-adjoint operators from $l^2_{comp}$ to $l^2_{loc}$. Besides, the number of poles of this meromorphic continuation in the lower half-plane $\{Im E<0\}$ is at most equal to $L$ (c.f.  \cite[Theorem 1.1]{klopp13}). This kind of results is an analogue in the discrete setting for meromorphic continuation of resolvents of partial differential operators (c.f. e.g. \cite{SZ91}).\\ 
Now, we define the \textit {resonances} of $H^{\nn}_L$, the main objet to study in the present paper, as the poles of the above meromorphic continuation. The resonance widths, the imaginary parts of resonances, play an important role in the large time behavior of wave packets, especially the resonances of the smallest width that give the leading order contribution (see \cite{SZ91} for an intensive study of resonances in the continuous setting and \cite{IK12, IK14, IK141, BNW05, korotyaev11} for a study of resonances of various $1$D operators).
\subsection{Resonance equation for the operator $H^{\nn}_{L}$}\label{Ss:equation}
Let $H_L$ be $H^{\mathbb N}_L$ restricted to $[0, L]$ with Dirichlet b.c. at $L$. Then, assume that 
\begin{itemize}
\item $\lambda_0 \leq \lambda_1\leq \ldots \leq \lambda_L$ are Dirichlet eigenvalues of $H_L$;
\item $a_k=a_k(L):=|\varphi_k(L)|^2$ where $\varphi_k$ is a normalized eigenvector associated to $\lambda_k$.
\end{itemize}
Then, resonances of $H^{\mathbb N}_L$ are solutions of the following equation (c.f. \cite[Theorem 2.1]{klopp13}): 
\begin{equation}\label{eq:resN}
S_L(E):=\sum_{k=0}^L \dfrac{a_k}{\lambda_k-E}=-e^{-i\theta(E)}, \qquad E= 2\cos\theta(E).
\end{equation} 
where the determination of $\theta(E)$ is chosen so that $\text{Im}\theta(E)>0$ and $\text{Re}\theta(E)\in (-\pi,0)$ when $\text{Im}E>0$.\\
Note that the map $E\mapsto \theta(E)$ can be continued analytically from $\cc^{+}$ to the cut plane $\cc\backslash ((-\infty, 2]\cup [2,+\infty))$ and its continuation is a bijection from $\cc\backslash ((-\infty, 2]\cup [2,+\infty))$ to $(-\pi, 0)+i\rr$. In particular, $\theta(E) \in (-\pi, 0)$ for all $E\in (-2, 2)$.\\
Taking imaginary parts of two sides of the resonance equation \eqref{eq:resN}, we obtain that    
\begin{equation}\label{eq:easyfact}
\text{Im} S_L(E)=\text{Im} E \sum_{k=0}^L \dfrac{a_k}{|\lambda_k-E|^2}=e^{\text{Im}E}\sin(\text{Re}\theta(E)).
\end{equation}
Note that, according to the choice of the determination $\theta(E)$, whenever $\text{Im} E>0$,\\ 
$\sin(\text{Re}\theta(E))$ is negative and $\text{Im}S_L(E)>0$. Hence, all resonances of $H^{\mathbb N}_L$ lie completely in the lower half-plane $\{\text{Im} E<0\}$.\\

The distribution of resonances of $H^{\nn}_L$ in the limit $L\rightarrow +\infty$ was studied intensively in \cite{klopp13}. 
All results proved in \cite{klopp13} assume that the real part of resonances are far from the boundary point of the spectrum $\Sigma_{\zz}$ and far from the point $\pm 2$, the boundary of the essential spectrum of $-\Delta$. By "far", we mean the distance between resonances and $\partial \Sigma_{\mathbb Z} \cup \{\pm2\} $ is bigger than a positive constant independent of $L$.\\
In the present paper, we are interested in phenomena which can happen for resonances whose real parts are near $\partial \Sigma_{\zz}$ but still far from $\pm 2$. To study resonances below compact intervals in $\ringz$, the interior of $\Sigma_{\zz}$, the author in \cite{klopp13} introduced an analytical method to simplify and resolve the equation \eqref{eq:resN} (c.f. \cite[Theorems 5.1 and 5.2]{klopp13}). Unfortunately, such a method was efficient inside $\Sigma_{\zz}$ but does not seem to work near $\partial \Sigma_{\zz}$. Hence, a different approach is thus needed to study resonances near $\partial \Sigma_{\zz}$.\\ 
We observe from \eqref{eq:resN} that the behavior of resonances is completely determined by spectral data $(\ld_{k})_{k}, (a_{k})_{k}$ of $H_{L}$. As pointed out in \cite{klopp13}, the parameters $a_k$ associated to $\lambda_k \in \ringz$ near a boundary point $E_{0} \in \partial \Sigma_{\zz}$ can have two different behaviors depending on the potential $V$: Either $a_k \asymp \frac{1}{L}$ or $a_{k}$ is much smaller, $a_k \asymp  \frac{|\lambda_k-E_0|}{L}$. Each case requires a particular approach for studying resonances and note that \textbf{the latter is the generic one}. \textit{In this paper, we deal with the non-generic case i.e. $a_{k} \asymp \frac{1}{L}$} (c.f. \cite{phong151} for the treatment of the generic case). 
\subsection{Description of resonances of $H^{\mathbb N}_L$ near $\partial \Sigma_{\zz}$ in the non generic case}
Let $E_{0} \in (-2, 2)$ be the left endpoint of the band $B_{i}$ of $\Sigma_{\zz}$ and $L=Np+j$ with $0\leq j \leq p-1$. We will study resonances in the domain $\cD= [E_0, E_0+\ve_1]-i[0, \ve_2]$ where $\ve_1 \asymp \ve^2$ and $\ve_2 \asymp \ve^5$ with $\ve>0$ small.\\ 
Note that, for eigenvalues $\lambda_{k} \in \ringz$ near $E_{0}$, we have $\lambda_{k} \asymp E_{0}+\frac{(k+1)^{2}}{L^{2}}$ (see Lemma \ref{L:asymptoticformula}). This leads us to make the rescaling $z=L^2(E-E_0)$ and track down rescaled resonances $z$ in the new region $\tilde{\mathcal D}= \tilde{\mathcal D_{\ve}}=[0, \ve_{1} L^2]-i[0, \ve_{2} L^2]$. Corresponding to this rescaling, we define rescaled eigenvalues $\til_k=L^2(\lambda_k-E_0)$ and $\tilde{a}_k=L a_k$. 
Then, the resonance equation \eqref{eq:resN} is rewritten as 
\begin{equation}\label{eq:4.0}
f_L(z):=\sum_{k=0}^L \frac{\tia_k}{\til_k -z}=-\frac{1}{L} e^{-i\theta\left(E_{0}+\frac{z}{L^{2}} \right)}.
\end{equation}
Our goal is to describe solutions of \eqref{eq:4.0} in the domain $\tilde{\cD}$. Let $(\lambda^i_\ell)_{\ell}$ with $\ell \in [0, n_{i, \ve}]$  be all (distinct) eigenvalues of $H_L$ belonging to $[E_0, E_0+\ve_{1}] \subset B_{i}$. Note that $n_{i, \ve} \asymp \ve L$ for $L$ large by Lemma \ref{L:asymptoticformula}. Here we use the (local) enumeration w.r.t. bands of $\Sigma_{\zz}$ to enumerate eigenvalues in the band $B_{i}$. We renumber the corresponding $a_{k}$ in the same way. Then, it suffices to study the rescaling resonance equation \eqref{eq:4.0} in each rectangle $\mathcal D^i_n:=[\til^i_n, \til^i_{n+1}]-i[0, \ve^5 L^2]$ with $0 \leq n \leq \ve L/C_{1}$ with $C_{1}>0$ large and in the rectangle $\mathcal R^i=[0, \til^i_0]-i[0, L^2\ve^{5}]$.\\
\textit{Lemma \ref{L:asymptoticformula} implies that, for all $\lambda_k  \in \mathring{B_i}$, the interior of $B_{i}$, and close to $E_{0}$, $\til_k \asymp (k+1)^2$ (we will use $\til_k \asymp k^2$ when $k\geq 1$). Moreover, from our assumption on $a_{k}$, the associated $\tia_k$ has the constant magnitude. Note that, in the non generic case, it is possible that $E_0 \in \sigma(H_L)$ for $L$ large. Then, according to our enumeration, $\til^i_0=0$ and $\tia^i_0$ is still of order $1$ (c.f. \cite[Remark 3.4]{phong151}).}

Here is our strategy to study resonances in the present case.\\ 
For $0\leq n \leq \ve L/C_1$ with $C_1>0$ large, we define $\Delta_n:=\frac{(n+1)}{\kappa(\ln (n+1)+1)}$ and $x_0=\til^{i}_{n+1}-\til^{i}_n$. First of all, we establish the subregions in $\cD^i_n$ and $\mathcal R^i$ which contain no resonances. They are rectangles greyed out in Figures \ref{F:f1}-\ref{F:f4}. The white regions $\Omega^{i}_n$, $\tilde{\Omega}^{i}_n$ and $\Omega^{i}$ in Figures \ref{F:f1}-\ref{F:f4} are the regions where we will study the existence and uniqueness of resonances. 
\begin{figure}[H]
	\begin{center}
		\begin{tikzpicture}[line cap=round,line join=round, x=0.6cm,y=0.6cm]
		\clip(0.29,-4.76) rectangle (12.17,3.55);
		\fill[color=gray,fill=gray,fill opacity=0.33] (2,2) -- (4,2) -- (4,0) -- (2,0) -- cycle;
		\fill[color=gray, fill=gray,fill opacity=0.33] (8,2) -- (10,2) -- (10,0) -- (8,0) -- cycle;
		\fill[color=gray,fill=gray,fill opacity=0.33] (2,-2) -- (1.99,-4.72) -- (10.01,-4.72) -- (10,-2) -- cycle;
		\draw [dash pattern=on 4pt off 4pt] (2,2)-- (4,2);
		\draw [dash pattern=on 4pt off 4pt] (4,2)-- (8,2);
		\draw (4,2)-- (8,2);
		\draw (4,2)-- (4,0);
		\draw (4,0)-- (2,0);
		\draw (2,-2)-- (2,0);
		\draw (8,2)-- (8,0);
		\draw (10,0)-- (8,0);
		\draw (10,0)-- (10,-2);
		\draw [dash pattern=on 4pt off 4pt] (10,-2)-- (2,-2);
		\draw [dash pattern=on 4pt off 4pt] (2,2)-- (2,0);
		\draw [dash pattern=on 4pt off 4pt] (8,2)-- (10,2);
		\draw [dash pattern=on 4pt off 4pt] (10,2)-- (10,0);
		\draw (10,-2)-- (2,-2);
		\draw (1.4,3.50) node[anchor=north west] {$ \tilde{\lambda}^{i}_n $};
		\draw (9.70,3.50) node[anchor=north west] {$ \tilde{\lambda}^{i}_{n+1} $};
		\draw (10.39,-1.01) node[anchor=north west] {$-\frac{x_0^2}{\varepsilon L} $};
		\draw (2.9,3.50) node[anchor=north west] {$ \tilde{\lambda}^{i}_n+\Delta_n $};
		\draw (6.4,3.50) node[anchor=north west] {$\tilde{\lambda}^{i}_{n+1}-\Delta_n $};
		\draw [<->, >=latex] (1.42,2) -- (1.42,0);
		\draw (0.71,1.63) node[anchor=north west] {$ \Delta_n $};
		\draw (5.6,0.38) node[anchor=north west] {$ \Omega^{i}_n $};
		\draw [] (2,2)-- (4,2);
		\draw [] (4,2)-- (4,0);
		\draw [] (4,0)-- (2,0);
		\draw [] (2,0)-- (2,2);
		\draw [] (8,2)-- (10,2);
		\draw [] (10,2)-- (10,0);
		\draw [] (10,0)-- (8,0);
		\draw [] (8,0)-- (8,2);
		\draw [] (2,-2)-- (1.99,-4.72);
		\draw [] (1.99,-4.72)-- (10.01,-4.72);
		\draw [] (10.01,-4.72)-- (10,-2);
		\draw [] (10,-2)-- (2,-2);
		\draw (10.39,-3.89) node[anchor=north west] {$-\varepsilon^5 L^2$};
		\draw (10.0,0.5) node[anchor=north west] {$-\Delta_n$};
		\begin{scriptsize}
		\fill [] (2,2) circle (1.5pt);
		\draw[] (2.27,2.46) node {};
		\fill [] (4,2) circle (1.5pt);
		\draw[] (4.25,2.35) node {$A$};
		\fill [] (8,2) circle (1.5pt);
		\draw[] (8.28,2.35) node {$B$};
		\fill [] (4,0) circle (1.5pt);
		\draw[] (4.25,0.45) node {$D$};
		\fill [] (2,0) circle (1.5pt);
		\draw[] (2.24,0.45) node {$E$};
		\fill [] (2,-2) circle (1.5pt);
		\draw[] (2.27,-1.53) node {$F$};
		\fill [] (8,0) circle (1.5pt);
		\draw[] (8.28,0.45) node {$C$};
		\fill [] (10,0) circle (1.5pt);
		\draw[] (10.29,0.45) node {$H$};
		\fill [] (10,-2) circle (1.5pt);
		\draw[] (10.19,-1.53) node {$G$};
		\fill [] (10,2) circle (1.5pt);
		\draw[] (1.75,2.46) node {};
		\draw[] (1.71,0.45) node {};
		\fill [] (1.99,-4.72) circle (1.5pt);
		\fill [] (10.01,-4.72) circle (1.5pt);
		\end{scriptsize}
		\end{tikzpicture}
		\caption[Figure1]{Resonance free region as $\Delta_n<\frac{x_0^2}{\ve L}$}
		\label{F:f1}
	\end{center}
\end{figure}

\begin{figure}[H]
	\begin{center}
		\begin{tikzpicture}[line cap=round,line join=round, x=0.6cm,y=0.6cm]
		\clip(-0.2,-2.91) rectangle (12.71,3.55);
		\fill[color=gray,fill=gray,fill opacity=0.3] (2,2) -- (4,2) -- (4,0) -- (2,0) -- cycle;
		\fill[color=gray,fill=gray,fill opacity=0.3] (8,2) -- (10,2) -- (10,0) -- (8,0) -- cycle;
		\fill[color=gray,fill=gray,fill opacity=0.3] (4,1) -- (4,0) -- (2,0) -- (2,-2) -- (10,-2) -- (10,0) -- (8,0) -- (8,1) -- cycle;
		\draw [dash pattern=on 1pt off 1pt] (4,2)-- (8,2);
		\draw (4,2)-- (8,2);
		\draw [dash pattern=on 3pt off 3pt] (4,2)-- (4,0);
		\draw [dash pattern=on 1pt off 1pt] (4,0)-- (2,0);
		\draw [dash pattern=on 1pt off 1pt] (8,2)-- (8,0);
		\draw [dash pattern=on 1pt off 1pt] (10,0)-- (8,0);
		\draw (8,2)-- (10,2);
		\draw (10,2)-- (10,0);
		\draw (1.60,3.50) node[anchor=north west] {$ \tilde{\lambda}^{i}_n $};
		\draw (9.75,3.50) node[anchor=north west] {$\tilde{\lambda}^{i}_{n+1} $};
		\draw (3.0,3.50) node[anchor=north west] {$ \tilde{\lambda}^{i}_n+\Delta_n $};
		\draw (6.5,3.50) node[anchor=north west] {$ \tilde{\lambda}^{i}_{n+1}-\Delta_n $};
		\draw [<->, >=latex] (1.42,2) -- (1.42,0);
		\draw (0.3,1.5) node[anchor=north west] {$\Delta_n $};
		\draw (4,1)-- (8,1);
		\draw (4,2)-- (4,1);
		\draw (8,2)-- (8,1);
		\draw (5.75,1.99) node[anchor=north west] {$\tilde{\Omega}^{i}_n$};
		\draw (10.38,1.78) node[anchor=north west] {$- \frac{x_0^2}{\varepsilon L} $};
		\draw (10.43,0.53) node[anchor=north west] {$-\Delta_n$};
		\draw (10.43,-1.45) node[anchor=north west] {$-\varepsilon^5 L^2$};
		\draw [dash pattern=on 3pt off 3pt] (2,1)-- (4,1);
		\draw (2,2)-- (4,2);
		\draw (2,2)-- (2,0);
		\draw (2,0)-- (4,0);
		\draw (4,0)-- (4,1);
		\draw [] (2,2)-- (4,2);
		\draw [] (4,2)-- (4,0);
		\draw [] (4,0)-- (2,0);
		\draw [] (2,0)-- (2,2);
		\draw [dash pattern=on 3pt off 3pt] (8,1)-- (10,1);
		\draw [] (8,2)-- (10,2);
		\draw [] (10,2)-- (10,0);
		\draw [] (10,0)-- (8,0);
		\draw [] (8,0)-- (8,2);
		\draw [] (4,1)-- (4,0);
		\draw [] (4,0)-- (2,0);
		\draw [] (2,0)-- (2,-2);
		\draw [] (2,-2)-- (10,-2);
		\draw [] (10,-2)-- (10,0);
		\draw [] (10,0)-- (8,0);
		\draw [] (8,0)-- (8,1);
		\draw [] (8,1)-- (4,1);
		\begin{scriptsize}
		\fill [] (4,2) circle (1.5pt);
		\draw[] (4.4,2.29) node {$A_1$};
		\fill [] (8,2) circle (1.5pt);
		\draw[] (7.7,2.29) node {$B_1$};
		\fill [] (4,0) circle (1.5pt);
		\fill [] (2,0) circle (1.5pt);
		\fill [] (8,0) circle (1.5pt);
		\fill [] (10,0) circle (1.5pt);
		\fill [] (10,2) circle (1.5pt);
		\fill [] (4,1) circle (1.5pt);
		\draw[] (4.4,1.29) node {$D_1$};
		\fill [] (8,1) circle (1.5pt);
		\draw[] (7.7,1.29) node {$C_1$};
		\fill [] (2,1) circle (1.5pt);
		\fill [] (10,1) circle (1.5pt);
		\fill [] (2,-2) circle (1.5pt);
		\fill [] (10,-2) circle (1.5pt);
		\fill [] (2,2) circle (1.5pt);
		\end{scriptsize}
		\end{tikzpicture}
		\caption[Figure2]{Resonance free resonance as $\Delta_n \geq \frac{x_0^2}{\ve L}$}
		\label{F:f2}
	\end{center}
\end{figure}
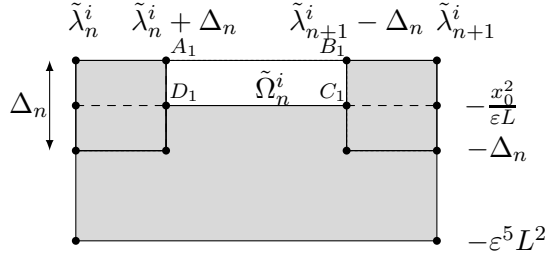

\begin{figure}[H]
\begin{center}
\begin{tikzpicture}[line cap=round,line join=round,x=0.6cm,y=0.6cm]
\clip(0.84,-3.33) rectangle (12.08,3.36);
\fill[color=gray,fill=gray,fill opacity=0.3] (8,2) -- (10,2) -- (10,0) -- (8,0) -- cycle;
\fill[color=gray,fill=gray,fill opacity=0.3] (8,1) -- (2,1) -- (2,-2) -- (10,-2) -- (10,0) -- (8,0) -- cycle;
\draw [dash pattern=on 1pt off 1pt] (8,2)-- (8,0);
\draw [dash pattern=on 1pt off 1pt] (10,0)-- (8,0);
\draw (8,2)-- (10,2);
\draw (10,2)-- (10,0);
\draw (1.7,3.3) node[anchor=north west] {$0$};
\draw (9.85,3.3) node[anchor=north west] {$\tilde{\lambda}^i_{0} $};
\draw (7.4,3.4) node[anchor=north west] {$\tilde{\lambda}^i_0-\delta_1$};
\draw (8,2)-- (8,1);
\draw (5.75,1.99) node[anchor=north west] {$\Omega^i$};
\draw (10.38,1.78) node[anchor=north west] {$- \frac{1}{\varepsilon L} $};
\draw (10.43,0.53) node[anchor=north west] {$-\delta_1$};
\draw (10.43,-1.45) node[anchor=north west] {$-\varepsilon^4 L^2$};
\draw [dash pattern=on 3pt off 3pt] (8,1)-- (10,1);
\draw (8,2)-- (10,2);
\draw  (10,2)-- (10,0);
\draw (10,0)-- (8,0);
\draw  (8,0)-- (8,2);
\draw (2,2)-- (8,2);
\draw (2,1)-- (8,1);
\draw (2,-2)-- (10,-2);
\draw (10,-2)-- (10,0);
\draw  (8,1)-- (2,1);
\draw  (2,1)-- (2,-2);
\draw  (2,-2)-- (10,-2);
\draw  (10,-2)-- (10,0);
\draw  (10,0)-- (8,0);
\draw  (8,0)-- (8,1);
\draw (2,2)-- (2,1);
\begin{scriptsize}
\fill (8,2) circle (1.5pt);
\draw (8.3,2.2) node {$B_3$};
\fill  (8,0) circle (1.5pt);
\fill  (10,0) circle (1.5pt);
\fill  (10,2) circle (1.5pt);
\fill  (8,1) circle (1.5pt);
\draw (8.3,1.29) node {$C_3$};
\fill  (2,1) circle (1.5pt);
\draw (2.3,1.29) node {$D_3$};
\fill  (10,1) circle (1.5pt);
\fill  (2,-2) circle (1.5pt);
\fill  (10,-2) circle (1.5pt);
\fill  (2,2) circle (1.5pt);
\draw (2.3,2.2) node {$A_3$};
\end{scriptsize}
\end{tikzpicture}
\caption[Figure4]{Resonance free region in $\mathcal R^i=[0, \til^i_0]-i[0, L^2\ve]$ }
\label{F:f4}
\end{center}
\end{figure}
Next, we give a description of resonances in $\Omega^{i}_{n}$. Note that this domain corresponds to the case $n\gtrsim \frac{L}{\ln L}$.
\begin{theorem}\label{T:abcd}
Assume that $n>\eta \frac{L}{\ln L}$ and put $x_0=\til^{i}_{n+1}-\til^{i}_n$.\\ 
Let $\Omega^i_n$ be the complement of two squares $[\til^{i}_n, \til^{i}_n+\Delta_n]+i[-\Delta_n, 0]$ and $ [\til^{i}_{n+1}, \til^{i}_{n+1}-\Delta_n]+i[-\Delta_n, 0]$ in the rectangle  $[\til^{i}_n, \til^{i}_{n+1}]+i\left[-\frac{x_0^2}{\ve L}, 0\right]$ ( the region $ABCHGFED$ in Figure \ref{F:f1}).\\
Then, there exists at least one rescaled resonance in $\Omega^i_n$. 
Hence, $|Im z| \lesssim \frac{n^2}{\ve L}$ for all resonances in $\Omega^i_n$.\\
Moreover, if 
 $-\frac{1}{L}e^{-i \theta(E_0)}$ belongs to $A'B'C'D'=f_L(ABCD)$, the rescaled resonance, says $z_n$, is unique and
 $$|\text{Im} z_n| \leq \Delta_n=\frac{n}{\kappa \ln n}\asymp \frac{n}{\kappa \ln L}\lesssim \frac{n^2}{\ve L}.$$ 
\end{theorem}
For $n$ smaller, our result is more satisfactory. We can be sure that there is one and only one resonance in $\tilde{\Omega}^i_n$.
\begin{theorem}\label{T:contour1} Pick $n<\frac{\eta L}{\ln L}$ with $\eta>0$ small.
Let $x_0=\til_{n+1}-\til_n$ and $\tilde{\Omega}^i_n$ be the rectangle $[\til_n+\Delta_n, \til_{n+1}-\Delta_n] +i\left[-\frac{x_0^2}{\ve L}, 0 \right]$ in Figure \ref{F:f2}.\\
Then, $f_L$ is bijective from $\tilde{\Omega}^i_n$ on $f_L(\tilde{\Omega}^i_n)$ and $|f'_L(z)| \gtrsim \frac{1}{n^2}$. Moreover, there exists a unique rescaled resonance $\tilde{z}_n$ in  $\tilde{\Omega}^i_n$. It satisfies, 
$$ |\text{Im} \tilde{z}_n| \lesssim \frac{n^2}{\ve L}.$$
\end{theorem}
Finally, there are no rescaled resonances in $\mathcal R^{i}$. 
\begin{theorem}\label{T:contour3}
Pick $0<\delta_1< \til_0$ and $\ve$ small, fixed numbers.\\
Let $E_0 \in (-2, 2)$ be the left endpoint of the $i$th band $B_i$ of $\Sigma_{\zz}$. Let $(\lambda^i_{\ell})_{\ell=0}^{n_i}$ be (distinct) eigenvalues of $H_L$ in $B_{i}$.\\
Let $\Omega^i$ be the rectangle $[0, \til^i_0-\delta_1]+i\left[-\frac{1}{\ve L}, 0\right]$ in Figure \ref{F:f4}.\\
Then, $f_L$ is bijective from $\Omega^i$ on $f_L(\Omega^i)$ and $|f'_L(z)| \geq c>0$. Moreover, 
$f_L(\Omega^i)$ does not contain the point $-\frac{e^{-i \theta(E_0)}}{L}$; hence, there are no resonances in  $\Omega^i$.
\end{theorem}
 From Theorems \ref{T:abcd}-\ref{T:contour3}, we figure out that $\text{Im}z$, the width of rescaled resonances, is always bounded by $\frac{n^{2}}{\ve L}$ up to a constant factor. Hence, when $n \asymp \ve L$ (far from $\partial \Sigma_{\zz}$), $|\text{Im}z|$ is smaller than $\ve L$ as pointed out in \cite{klopp13}. However, when $n$ is small (close to $\partial \Sigma_{\zz}$), the width of rescaled resonances is much smaller. It can be smaller than $\frac{1}{L^{3}}$ up to a constant factor.
 
The basic idea to prove Theorems \ref{T:abcd}-\ref{T:contour3} is to simplify the rescaled resonance equation \eqref{eq:4.0} as much as possible via Rouch\'{e}'s theorem as we did for the generic case (c.f. \cite{phong151}). In the generic case, when we approximate the sum $S_{L}(E)$ in a region close to the eigenvalue $\ld^{i}_{n}$, we can simply keep the term $\frac{a^{i}_{n}}{\ld^{i}_{n}-E}$ and replace the remaining sum by some appropriate number which is independent of $E$. Consequently, we obtain a very simple but efficient approximation by Rouch\'{e}'s theorem. This enables us to obtain a detailed description of resonances near $\partial \Sigma_{\zz}$. However, in the present case, the situation is worse. In a domain near $\til^{i}_{n}$, if we keep the term $\frac{a^{i}_{n}}{\ld^{i}_{n}-E}$ or two terms $\frac{a^{i}_{n}}{\ld^{i}_{n}-E}$ and $\frac{a^{i}_{n+1}}{\ld^{i}_{n+1}-E}$ or even more, we can not find a suitable approximation for the sum of the other terms. We always have to deal with the situation that the error of our approximations for $f_{L}(z)$ can not be good enough to apply Rouch\'{e}'s theorem. So, the solution which we come up with is the following. We will only use Rouch\'{e}'s theorem to replace RHS of \eqref{eq:4.0} by the number $-\frac{1}{L} e^{-i\theta(E_{0})}$. Instead of approximating $f_{L}(z)$, we will study explicitly the images of domains where we want to track down resonances under $f_{L}(z)$. Then, based on the shape of the images of involved domains through $f_{L}$ as well as the relative position between them and the point $-\frac{1}{L} e^{-i\theta(E_{0})} \in \cc$, we can obtain the information on resonances.
 
Our paper is organized as follows. First of all, Section \ref{S:sd} is a recall of the behavior of spectral data of $H_{L}$ introduced in \cite{klopp13, phong151}. Next, in Section \ref{S:free}, we prove the free resonance regions in $\cD^i_n$ and $\mathcal R^i$. Finally, Section \ref{Ss:closest} is dedicated to the proofs of Theorems \ref{T:abcd}-\ref{T:contour3}.

\textbf{Notations:} Throughout the present paper, we will write $C$ for constants whose values can vary from line to line. Constants marked $C_{i}$ are fixed within a given argument. We write $a\lesssim b$ if there exists some $C>0$ independent of parameters coming into $a, b$ s.t. $a \leq Cb$. Finally, $a\asymp b$ means $a\lesssim b$ and $b\lesssim a$. 

\section{Spectral data near the boundary of $\Sigma_{\zz}$} \label{S:sd}
From \eqref{eq:resN}, resonances of $H_{L}^{\nn}$ depend only on the spectral data of the operator $H_{L}$ i.e., the eigenvalues and corresponding normalized eigenvectors of $H_{L}$. In order to "resolve" the resonance equation \eqref{eq:resN}, it is essential to understand how eigenvalues of $H_L$ behave and what the magnitudes of $a_l:=|\varphi_l(L)|^2$ are in the limit $L \rightarrow +\infty$.\\  
Before stating the properties of spectral data of $H_L$, one defines the \textit{quasi-momentum} of $H^{\mathbb Z}$:\\
Let $V$ be a periodic potential of period $p$ and $L$ be large. For $0\leq k\leq p-1$, one defines $\widetilde{T}_k=\widetilde{T}_k(E)$ to be a monodromy matrix for the periodic finite difference operators $H^{\mathbb Z}$, that is,
\begin{equation}\label{eq:monodromy}
\widetilde{T}_k(E)= T_{k+p-1, k}(E)=T_{k+p-1}(E)\ldots T_k(E)=
\begin{pmatrix}
 a^k_p(E) &a^k_p(E)\\
 a^k_{p-1}(E) &a^k_{p-1}(E)
\end{pmatrix}   
\end{equation}   
where $\{T_l(E)\}$ are transfer matrices of $H^{\mathbb Z}$: 
\begin{equation}\label{eq:transfer}
T_l(E)= 
\begin{pmatrix}
 E-V_l &-1\\
 1 &0
\end{pmatrix}.   
\end{equation}   
Besides, for $k\in \{0, \ldots, p-1\}$ we write
\begin{equation}\label{eq:product}
T_{k-1}(E\ldots T_0(E)= \begin{pmatrix}
a_k(E) & b_k(E)\\
 a_{k-1}(E) & b_{k-1}(E)
\end{pmatrix}.
\end{equation}
We observe that the coefficients of $\widetilde{T}_k(E)$ are monic polynomials in $E$. Moreover, $a^k_p(E)$ has degree $p$ and $b^k_p(E)$ has a degree $p-1$. The determinant of $T_l(E)$ equals to $1$ for any $l$, hence, $\det \widetilde{T}_k(E)=1$. Besides, $k \mapsto \widetilde{T}_k(E)$ is $p-$periodic since $V$ is a $p-$periodic potential. Moreover, for $j<k$
$$\widetilde{T}_k(E)= T_{k,j}(E) \widetilde{T}_j(E) T^{-1}_{k,j}(E).$$
Thus the discriminant $\Delta(E):= \text{tr} \widetilde{T}_k(E)=a^k_p(E)+b^k_{p-1}(E)$ is independent of $k$ and so are $\rho(E)$ and $\rho(E)^{-1}$, eigenvalues of $\widetilde{T}_k(E)$. Now, one can define the Floquet \textit{quasi-momentum} $E \mapsto \theta_p(E)$ by 
\begin{equation}\label{eq:2.5}
\Delta(E)=\rho(E)+\rho^{-1}(E)= 2\cos\left(p \theta_p(E)\right).
\end{equation}
Then, one can show that the spectrum of $H_{\nn}$, $\Sigma_{\mathbb Z}$, is the set $\{ E | |\Delta(E)|\leq 2 \}$ and $$\partial \Sigma_{\mathbb Z}=\{ E | |\Delta(E)|=2 \text{ and $\widetilde{T}_{0}(E)$ is not diagonal}\}.$$ 
Note that each point of $\partial \Sigma_{\zz}$ is a branch point of $\theta_p(E)$ of square-root type.
\begin{figure}[H]
\begin{center}
\begin{tikzpicture}[line cap=round,line join=round,x=0.67 cm,y=0.67 cm]
\clip(-2.76,-2.14) rectangle (6.47,4.08);
\draw[smooth,samples=100,domain=-2.758029237549914:6.47493127408347] plot(\x,{2*sin(((\x)+1.5)*180/pi)+0.09});
\draw (-4,2)-- (-4,0);
\draw [domain=-2.76:6.47] plot(\x,{(--2-0*\x)/2});
\draw (-4,0)-- (-4,-2);
\draw [domain=-2.76:6.47] plot(\x,{(-2-0*\x)/2});
\draw [line width=2.8pt] (-2.07,0)-- (-1.06,0);
\draw [line width=2.8pt] (1.17,0)-- (2.22,0);
\draw [line width=2.8pt] (4.21,0)-- (5.26,0);
\draw[->, >=latex] (0,-6)-- (0, 4.08);
\draw[->, >=latex](-6,0)-- (9,0);
\draw [->, >=latex] (5.26,0) -- (6.47,0);
\draw (0.06,4.21) node[anchor=north west] {$ \Delta(E) $};
\draw (5.71,0.04) node[anchor=north west] {$ E $};
\draw (0.11,1.69) node[anchor=north west] {$ 2 $};
\draw (-0.2,-0.88) node[anchor=north west] {$-2$};
\begin{scriptsize}
\draw (-3.82,-0.68) node {$E$};
\fill [] (-1.06,0) circle (1.5pt);
\draw[] (-0.86,0.33) node {$E^{+}_1$};
\fill [] (-2.07,0) circle (1.5pt);
\draw[] (-1.94,0.33) node {$E^{-}_1$};
\fill [] (1.17,0) circle (1.5pt);
\draw[] (1.36,0.33) node {$E^{-}_2$};
\fill [] (2.22,0) circle (1.5pt);
\draw[] (2.4,0.33) node {$E^{+}_2$};
\fill [] (4.21,0) circle (1.5pt);
\draw[] (4.4,0.33) node {$E^{-}_3$};
\fill [] (5.26,0) circle (1.5pt);
\draw[] (5.46,0.33) node {$E^{+}_3$};
\fill [] (0,0) circle (1.5pt);
\draw[] (0.2,0.33) node {$O$};
\end{scriptsize}
\end{tikzpicture}
\caption[Figure2]{Function $\Delta(E)$.}
\label{F:fk2}
\end{center}
\end{figure}
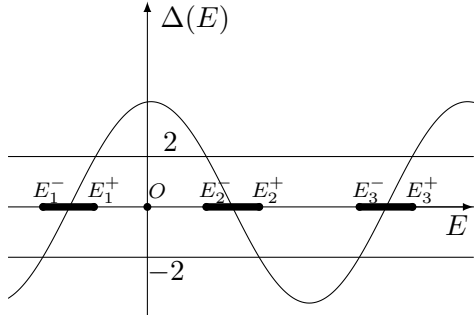
One decomposes $\Sigma_{\zz}$ into its connected components i.e.$\Sigma_{\zz}=\bigcup\limits_{i=1}^q B_i$ with $q<p$. Let $c_{i}$ be the number of closed gaps contained in $B_{i}$. Then, $\theta_{p}$ maps $B_{i}$ bijectively into $\sum\limits_{\ell=1}^{i-1} (1+c_{\ell}) \frac{\pi}{p} + \frac{\pi}{p}[0, c_{i}]$. Moreover, on this set, the derivative of $\theta_{p}$ is proportional to the common density of states $n(E)$ of $H^{\mathbb Z}$ and $H^{\mathbb N}$: 
$$ \theta'_p(E)= \pi n(E).$$
One has the following description for spectral data of $H_L$.
\begin{theorem}\cite[Theorem 4.2]{klopp13}\label{T:spectraldata}   
For any $j \in \{0, \ldots, p-1\},$ there exists $h_j: \Sigma_{\zz} \rightarrow \rr$ , a continuous function that is real analytic in a neighborhood of $\ringz$ such that, for $L=Np+j$,
\begin{enumerate}
\item The function $h_j$ maps $B_{i}$ into $\left(-(c_{i}+1)\pi, (c_{i}+1)\pi \right)$ where $c_i$ is the number of closed gaps in $B_i$;
\item the function $\theta_{p, L}= \theta_p - \frac{h_j}{L-j}$ is strictly monotonous on each band $B_{i}$ of $\Sigma_{\zz}$;
\item for $1\leq i \leq q$, the eigenvalues of $H_L$ in $B_{i}$, the $i$th band of $\Sigma_{\zz}$, says $(\lambda^{i}_k)_k$ are the solutions (in $\Sigma_{\zz}$) to the quantization condition
\begin{equation}\label{eq:quan1}
\theta_{p, L}(\lambda^i_k) =\dfrac{k \pi}{L-j}, \quad k\in \zz.
\end{equation}
\item If $\lambda$ is an eigenvalue of $H_L$ outside $\Sigma_{\zz}$ for $L=Np+j$ large, there exists $\lambda_{\infty} \in \Sigma^{+}_0 \cup \Sigma^{-}_j \backslash \Sigma_{\zz}$ s.t. $|\lambda-\lambda_{\infty}| \leq e^{-cL}$ with $c>0$ independent of $L$ and $\lambda$. 
\end{enumerate}
\end{theorem}
When solving the equation \eqref{eq:quan1}, one has to do it for each band $B_{i}$, and for each band and each $k$ such that $\frac{k \pi}{L-j} \in \theta_{p,L}(B_{i})$, \eqref{eq:quan1} admits a unique solution. But, it may happens that one has two solutions to \eqref{eq:quan1} for a given $k$ belonging to neighboring bands.
\begin{remark} \label{R:nearboundary}
From \cite[Section 4]{klopp13}, we have the following behavior of $a_k$ associated to $\lambda_k$ which is close to $\partial \Sigma_{\zz}$.\\ Let $E_0 \in \partial \Sigma_{\zz}$ and $L=Np+j$. We define $d_{j+1}=a_{j+1}(E_0) (a^0_p(E_0) -\rho^{-1}(E_0))+b_{j+1}(E_0) a^0_{p-1}(E_0)$ where $a_{j+1}, b_{j+1}, a^0_p, a^0_{p-1}$ are polynomials defined in \eqref{eq:monodromy} and \eqref{eq:product}. Then, one distinguishes two cases:    
\begin{enumerate}
	\item If $a^0_{p-1}(E_0)=0$, then 
	$$a_k=|\varphi_k(L)|^2 \asymp \frac{|\lambda_k-E_0|}{L-j} \text{ and } |\varphi_k(0)|^2 \asymp \frac{1}{L-j}.$$
	\item If $a^0_{p-1}(E_0)\ne 0$, then
	\begin{itemize}
		\item if $d_{j+1}\ne 0$, one has 
		$$|\varphi_k(L)|^2 \asymp \frac{|\lambda_k-E_0|}{L-j} \text{ and } |\varphi_k(0)|^2 \asymp \frac{|\lambda_k-E_0|}{L-j}.$$
		\item if $d_{j+1}=0$, one has
		$$|\varphi_k(L)|^2 \asymp \frac{1}{L-j} \text{ and } |\varphi_k(0)|^2 \asymp \frac{|\lambda_k-E_0|}{L-j}.$$	
	\end{itemize}	
\end{enumerate}
Besides, according to \cite{phong151}, if $E_{0} \in \sigma(H_{L})$ for $L$ large, says $E_{0}=\ld_{k}$, then the associated $a_{k}$ is of order $\frac{1}{L}$. 
\end{remark}
Finally, we would like to remind readers of the behavior of eigenvalues of $H_L$ close to $E_0$.   
\begin{lemma} \cite[Lemma 3.5]{phong151} \label{L:asymptoticformula}
Let $E_0 \in (-2, 2)$ be the left endpoint of the $i$th band $B_i$ of $\Sigma_{\zz}$. Let 
$\lambda^{i}_0 < \lambda^{i}_1< \ldots < \lambda^{i}_{n_i}$ be eigenvalues of $H_{L}$ in $\mathring{B_i}$, the interior of $B_{i}$.\\
Pick $\ve>0$ a small, fixed number and $\ve_1\asymp \ve^2$. Let $I=I_{\ve_1}: =[E_0, E_0+\ve_1] \subset (-2,2) \cap \Sigma_{\mathbb Z}$.\\
Assume that $\lambda^{i}_k$ is an eigenvalue of $H_L$ in $I$. 
Then, $k\leq \ve (L-j)$ and $\lambda^{i}_k-E_0 \asymp \frac{(k+1)^2}{L^2}$ (for $k\geq 1$, we will write $\lambda^{i}_k-E_0 \asymp \frac{k^2}{L^2}$ instead).\\ 
Moreover, there exists $\alpha>0$ s.t. for any $0 \leq n <k \leq \ve (L-j)$, we have
\begin{equation}\label{eq:survive}
\frac{|k^2-n^2|}{\alpha L^2}  \leq |\lambda^{i}_k -\lambda^{i}_n| \leq \frac{\alpha |k^2-n^2|}{L^2}.
\end{equation}
\end{lemma}
\begin{proof}[Proof of Lemma \ref{L:asymptoticformula}]
To simplify notations, we will skip the superscript $i$ in $\lambda^i_k$ of $H_L$ throughout this proof.\\ 
First of all, from the property of $\theta_{p}$ and $h_{j}$ near $E_{0}$, we have, for any $E$ near $E_{0}$, 
\begin{equation}\label{eq:114}
\theta_{p,L}(E)-\theta_{p,L}(E_{0})= c(L)\sqrt{|E-E_{0}|} \left( 1+o(1)\right) 
\end{equation}
where $|c(L)|$ is lower bounded and upper bounded by positive constants independent of $L$.\\
Put $L=Np+j$ where $p$ is the period of the potential $V$ and $0 \leq j \leq p-1$. According to Theorem \ref{T:spectraldata}, $\theta_{p,L}(E)$ is strictly monotone on $B_{i}$. W.o.l.g., we assume that $\theta_{p,L}(E)$ is strictly increasing on $B_{i}$.
Note that, in this lemma, we enumerate eigenvalues $\ld_{\ell}$ in $\mathring{B_{i}}$ with the index $\ell$ starting from $0$. Then, we have to modify the quantization condition \eqref{eq:quan1} in Theorem \ref{T:spectraldata} appropriately. Recall that the quantization condition is $\theta_{p,L}(\ld_{\ell})= \frac{\pi \ell}{L-j}$ where $\frac{\pi \ell}{L-j} \in \theta_{p,L}(B_{i})$ with $\ell \in \zz$. Assume that $\theta_{p}(E_{0})= \frac{m \pi}{p}$ with $m \in \zz$.\\ 
Put $\ell=\lambda N + \tilde{k}$ where $\lambda \in \zz$ and $0 \leq \tilde{k} \leq N-1$. 
We find $\lambda, \tilde{k}$ such that 
\begin{equation}\label{eq:20.1}
\frac{\ell \pi}{Np} - \theta_{p,L}(E_{0})= (\ld-m)\frac{\pi}{p} + \frac{\tilde{k}\pi +h_{j}(E_{0})}{Np} > 0.
\end{equation}
It is easy to see that, for $N$ large, the necessary condition is $\ld-m \geq -1$. Consider the case $\ld -m =-1$. Then, \eqref{eq:20.1} yields 
\begin{equation}\label{eq:20.2}
\tilde{k} \pi + h_{j}(E_{0}) > N \pi.
\end{equation}  
According to \cite[Lemma 4.7]{klopp13}, $h_{j}(E_{0}) \in \frac{\pi}{2} \zz$. We observe that if $h_{j}(E_{0})<0$, there does not exist $0 \leq \tilde{k} \leq N-1$ satisfying \eqref{eq:20.2}. Hence, $h_{j}(E_{0}) \in \frac{\pi}{2} \nn$. We distinguish two cases. First of all, assume that $h_{j}(E_{0}) \in \pi \nn$. Then, the first $\ell$ verifying \eqref{eq:20.1} and $\ld_{\ell} \in \ringz$ is $\ell_{0} = \frac{Np}{\pi} \theta_{p,L}(E_{0})+1$. Next, consider the case $h_{j}(E_{0}) \in \frac{\pi}{2}+ \pi \nn$. Then, the first $\ell$ chosen is $\ell_{0}=\frac{Np}{\pi} \theta_{p,L}(E_{0})+\frac{1}{2}$. Put $\ell_{k}= \ell_{0}+k$ and we associate $\ell_{k}$ to $\lambda_{k}$, the $(k+1)-$th eigenvalue in $\mathring{B_{i}}$. Then, we always have 
\begin{equation}\label{eq:shift}
\theta_{p,L}(\ld_{k})-\theta_{p,L}(E_{0})= \frac{(k+1)\pi}{L-j} + \frac{c_{0}}{L-j}
\end{equation}   
where $c_{0}=0$ if $h_{j}(E_{0}) \in \pi \zz$ and $c_{0}=-\frac{\pi}{2}$ otherwise.\\
Hence, \eqref{eq:114} and \eqref{eq:shift} yield $\ld_{k}-E_{0} \asymp \frac{(k+1)^{2}}{L^{2}}$ for all $\ld_{k} \in I$ with $\ve$ small and $L$ large. Consequently, $k \lesssim \ve (L-j)$.\\
Finally, we will prove the inequality \eqref{eq:survive}.\\
Recall that the functions $\theta_{p}(E_{0}+x^{2})$ and $h_{j}(E_{0}+x^{2})$ are analytic in $x$ on the whole band $B_{i}$. Then, we can expand these functions near $0$ to get 
$$\theta_{p}(E_{0}+x^{2}) = \theta_{p, 0}+ \theta_{p, 1}x +\theta_{p,2}x^{2} + O(x^{3}) \text{ where $\theta_{p, 0}=\theta_{p}(E_{0})$ and $\theta_{p, 1} \ne 0$}; $$
$$h_{j}(E_{0}+x^{2})= h_{j,0} +h_{j,1}x +h_{j,2}x^{2}+O(x^{3}) \text{ where $h_{j,0}=h_{j}(E_{0})$}.$$
Put $x_{k}=\sqrt{\ld_{k}-E_{0}}$. We can assume that $\theta_{p,L}$ is increasing on $B_{i}$. Then, \eqref{eq:shift} and the above expansions yield 
\begin{equation}\label{eq:300}
\theta_{p,1}(L) x_{k} + \theta_{p,2}(L) x^{2}_{k} +O(x_{k}^{3}) =\frac{(k+1) \pi}{L-j}+\frac{c_{0}}{L-j}
\end{equation}
where $\theta_{p, m}(L)=\theta_{p,m}- \frac{h_{j,m}}{L-j}$ for all $m\in \nn$, $c_{0}=0$ if $h_{j}(E_{0}) \in \pi \zz$ and $c_{0}=-\frac{\pi}{2}$ otherwise. Note that $|\theta_{p,1}(L)|$ is lower bounded and upper bounded by positive constants independent of $L$.\\
W.ol.g., assume that $c_{0}=0$. Then, we have 
\begin{equation}\label{eq:301}
x_{k}=  \tilde{c}(L) \cdot \frac{(k+1)\pi}{(L-j)} \cdot \frac{1}{1+ x_{k}g_{L}(x_{k})}  
\end{equation}
where $|\tilde{c}(L)|$ is lower bounded and upper bounded by positive constants independent of $L$. Moreover, the function $g_{L}$ is analytic near $0$; $g_L$ and its derivative are bounded near $0$ by constants independent of $L$.\\
Let $0\leq n<k \leq \ve (L-j)$, the equation \eqref{eq:301} yield 
\begin{equation}\label{eq:302}
x_{k}-x_{n}= \tilde{c}(L) \cdot \frac{\pi (k-n)}{L-j}\cdot \frac{1}{1+x_{k}g_{L}(x_{k})}+ \tilde{c}(L)\cdot \frac{n+1}{L}  \cdot \frac{x_{n}g_{L}(x_{n}) - x_{k}g_{L}(x_{k})}{\left(1+x_{k}g_{L}(x_{k}) \right) \left(1+x_{n}g_{L}(x_{n}) \right)}. 
\end{equation}
Note that the second term of the right hand side (RHS) of \eqref{eq:302} is bounded by $\ve |x_{k}-x_{n}|$ up to a constant factor. Hence, there exists a constant $C$ such that, for all $n<k\leq \ve(L-j)$, 
\begin{equation}\label{eq:303}
\frac{1}{C} \cdot \frac{k-n}{L-j} \leq |x_{k}-x_{n}| \leq C\cdot \frac{k-n}{L-j}. 
\end{equation}
On the other hand, $x_{k} \asymp \frac{k+1}{L-j}$ and $x_{n}\asymp \frac{n+1}{L-j}$. We thus have 
$|\ld_{k}-\ld_{n}|=|x^{2}_{k}-x^{2}_{n}| \asymp \frac{|k^{2}-n^{2}|}{L^{2}}$ for all $0 \leq n<k \leq \ve (L-j)$.
\end{proof}
\begin{remark}\label{R:spacing}
For $L$ large, the average  distance between two consecutive, distinct eigenvalues (the spacing) is $\frac{1}{L}$. Lemma \ref{L:asymptoticformula} says that, the spacing between eigenvalues near $\partial \Sigma_{\zz}$ is much smaller, the distance  between $\lambda^{i}_k$ and $\lambda^{i}_{k+1} \in I=[E_0, E_0+\ve_{1}]$ where $\ve_{1} \asymp \ve^{2}$ has magnitude $\frac{k+1}{L^2}$. This fact implies that the number of eigenvalues in the interval $I$ is asymptotically equal to $\ve L$ as $L \rightarrow +\infty$.
\end{remark}
\section{Resonance free regions}\label{S:free}
First of all, we state and prove the following lemma which will be useful for estimating $f_{L}(z)$.
\begin{lemma} \label{L:outsider1}
Pick $\eta>0$ and $E_{0}\in \partial \Sigma_{\zz}$. For $E\in J:=\left[E_{0}, E_{0}+\eta \right]+i \rr$, we define $z=L^{2}(E-E_{0})$ and $f_{\text{out}}(z) = \sum\limits_{|\lambda_{k}-E_{0}|> 2\eta} \frac{\tia_{k}}{\til_{k}-z}$
 Then, 
\begin{equation}\label{eq:outsider3}
|f_{\text{out}}(z)| \leq \frac{1}{\eta L} \text{ and } |\text{Im} f_{\text{out}}(z)| \leq \frac{|\text{Im} z|}{\eta^{2} L^{3}}
\end{equation}
and
\begin{equation}\label{eq:outsider4}
0< f'_{\text{out}}(z) \leq \frac{1}{\eta^{2} L^{3}} \text{ for all $E \in \left[E_{0}, E_{0} + \eta\right]$ }. 
\end{equation}
\end{lemma}
\begin{proof}[Proof of Lemma \ref{L:outsider1}] 
Note that $|\til_{k}-z| > \eta L^{2}$ for all $|\lambda_{k}-E_{0}|>2\eta$ and $E\in J$. On the other hand, 
$\text{Im} f_{\text{out}}(z) = \text{Im}z \sum\limits_{|\lambda_{k}-E_{0}|>2\eta} \frac{\tia_{k}}{|\til_{k}-z|^{2}}$ and 
$f'_{\text{out}}(z)=\sum\limits_{|\lambda_{k}-E_{0}|>2\eta} \frac{\tia_{k}}{(\til_{k}-z)^{2}}$. Hence, the claim follows.
\end{proof}
\subsection{Near the poles of \lowercase{$f_{\uppercase{l}}(z)$ }}\label{Ss:nearpoles}
Let $\lambda^{i}_0 < \lambda^{i}_1< \ldots < \lambda^{i}_{n_i}$ be eigenvalues of $H_{L}$ in $B_{i}$. With that enumeration, it is possible that $\ld^i_0 =E_0$ for all $L$ large. Note that if that case happens, the formula \eqref{eq:survive} in Lemma \ref{L:asymptoticformula} still holds for any pair $\ld^{i}_{0}, \ld^{i}_{k}$ with $0< k \lesssim \ve L$. In fact, $E_{0}$ is an eigenvalue of $H_{L}$ or not will not affect our results at all.\\
For each $0 \leq n \leq \ve L/C_{1}$ with $C_{1}>0$ large, the rectangle $\cD_n$ contains $\til_n, \til_{n+1}$, two poles of the meromorphic function $f_L(z)$. Since the modulus of $f_L(z)$ is big near these points, there are no resonances in those regions. Following is a quantitative version of this observation. 
\begin{lemma} \label{L:modulus} Let $E_0 \in (-2, 2)$ be the left endpoint of the band $B_{i}$ of $\Sigma_{\zz}$. Assume that 
$(\lambda^{i}_{\ell})_{\ell}$ with $0\leq \ell \leq n_{i}$ are (distinct) eigenvalues of $H_{L}$ in $B_{i}$.
Put $I=\left[E_0, E_0+\ve_{1} \right] \subset B_{i}$ where $\ve_{1} \asymp \ve^{2}$ with $\ve>0$ small. 
For each $0 \leq n \leq \ve L/C_1$ with $C_1>0$ large, we define
\begin{equation}\label{eq:define}
f_{n,L}(z):=\frac{\tia^{i}_n}{\til^{i}_n -z}+\frac{\tia^{i}_{n+1}}{\til^{i}_{n+1} -z}; \quad \tilde{f}_{n,L}(z):=f_L(z)-f_{n,L}(z)
\end{equation}
where $z=L^2(E-E_0)$ with $E\in I -i[0, \ve^{5}]$;\\ 
$\Delta_n:=\frac{(n+1)}{\kappa(\ln (n+1)+1)}$ where $\kappa$ is a large constant.\\
Then, 
\begin{itemize}
\item[$\bullet$] $|\tilde{f}_{n,L}(z)| \lesssim \frac{\ln (n+1)+1}{n+1} $ for all $z \in [\til^{i}_n, \til^{i}_{n+1}]+i\rr$,
\item[$\bullet$] $\tilde{f}'_{n,L}(z) \asymp \frac{1}{(n+1)^2}$ if $z$ is real and $z\in [\til^{i}_n, \til^{i}_{n+1}] $,
\item[$\bullet$] $|Im \tilde{f}_{n,L}(z)| \lesssim \frac{|Im z|}{n^2}$.
\end{itemize}
Consequently, for all $z \in  \left( [\til^{i}_n, \til^{i}_n\pm \Delta_n] \cap [0, \ve_1 L^2] \right) -i[0, \Delta_n]$, 
\begin{equation}\label{eq:4.-2}
|f_L(z)| \gtrsim \frac{1}{\Delta_n}\gtrsim \frac{1}{\ve L}.
\end{equation}
Note that, in the definition of $\Delta_n$, we choose $\kappa$ to be large so that $\til^{i}_{n}-\Delta_n>0$. Besides, $[\til^{i}_n, \til^{i}_n\pm \Delta_n]$ always belongs to $[0, \ve_{1}L^{2}]$ unless $n=0$ and $\til^{i}_{0}=0$ i.e. $E_{0} \in \sigma(H_{L})$ for any $L$ large.
\end{lemma}
\begin{proof}[Proof of Lemma \ref{L:modulus}]
We can choose $C_{1}>0$ large enough such that $\ld^{i}_{n} < E_0+\ve_1$ and $\lambda^{i}_{k}>E_{0}+2\ve_{1}$ if $k>\ve L$ and $\lambda^{i}_{k} \in B_{i}$. Then, Lemma \ref{L:outsider1} yield
\begin{equation}\label{eq:4.60}
\sum_{\ld_{k} \notin [E_{0}, E_{0}+2\ve_{1}]}\frac{\tia_k}{|\til_k -z|} \leq  \frac{1}{\ve_{1} L} \lesssim \frac{\ln (n+1)+1}{n+1}.
\end{equation}
Hence, it suffices to prove the same bound for the sum $S$ where 
\begin{equation*}\label{eq:4.1}
S=\sum_{k=0}^{n-1} \frac{\tia^{i}_k}{\til^{i}_k -z} + \sum_{k=n+2}^{\ve L} \frac{\tia^{i}_k}{\til^{i}_k -z}=:S_1+S_2.
\end{equation*}
Throughout the rest of the proof, we will omit the superscript $i$ to lighten the notation.\\
Recall that, by Lemma \ref{L:asymptoticformula}, $|\til_{k}-\til_{n}| \asymp |k^{2}-n^{2}|$ for all $k \ne n \in \left[0, \ve L/C_1\right]$. Hence,
\begin{align}\label{eq:115}
|S_1| &\leq  \sum_{k=0}^{n-1}\frac{\tia_k}{|\til_k -z|} \leq  \sum_{k=0}^{n-1}\frac{\tia_k}{|\til_n -\til_k|}\lesssim
\sum_{k=0}^{n-1}\frac{1}{(n-k)(n+k)}\notag\\ 
&\lesssim \sum_{k=1}^n \frac{1}{k(2n-k)} \lesssim \frac{\ln (n+1) +1}{n+1}.
\end{align}
Next, we will estimate the sum $S_2$.
\begin{align} \label{eq:farboundary}
|S_2|& \leq \sum_{k= n+2}^{\ve L}\frac{\tia_k}{|\til_k -z|}\lesssim \sum_{k\geq n+2} \frac{1}{k^2-(n+1)^2} \lesssim \sum_{k\geq 1} \frac{1}{k(k+2n+2)} \notag\\
& \lesssim  \frac{1}{2n+2}\sum_{k=1}^{2n+2}\frac{1}{k}\lesssim \frac{\ln (n+1)+1}{n+1}.
\end{align}
Hence, \eqref{eq:4.60}-\eqref{eq:farboundary} yield $| \tilde{f}_{n,L}(z)| \lesssim \frac{\ln (n+1)+1}{n+1}$.\\
Now, we will prove the second item of Lemma \ref{L:modulus}. Assume that $z$ is real and $z\in [\til_n, \til_{n+1}]$. Then, by Lemma \ref{L:outsider1}, we have
\begin{align}\label{eq:-1}
\tilde{f}'_{n,L}(z) &\leq \sum_{\substack{k\leq \ve L\\ k\ne n, n+1} }\frac{\tilde{a}_k}{(\til_k-z)^2}+ \frac{1}{\ve_1^2 L^3} \notag \\
&\leq  \sum_{k=0}^{n-1}\frac{\tilde{a}_k}{(\til_n-\til_k)^2}+\sum_{k=n+2}^{\ve L}\frac{\tilde{a}_k}{(\til_k-\til_{n+1})^2}
+ \frac{1}{\ve_1^2 L^3} \\ 
&\lesssim \frac{1}{(n+1)^2}+ \frac{1}{\ve_1^2 L^3} \lesssim \frac{1}{(n+1)^2}. \notag
\end{align}
On the other hand, for $z \in [\til_n, \til_{n+1}]$ and $n\geq 1$, we have
\begin{align}\label{eq:-2}
\tilde{f}'_{n,L}(z) &\geq \sum_{k=0}^{n-1}\frac{\tilde{a}_k}{(\til_n-\til_k)^2}\gtrsim  \sum_{k=0}^{n-1} \frac{1}{(n-k)^2(n+k)^2}\notag\\ 
&\gtrsim \sum_{k=1}^n \frac{1}{k^2(2n-k)^2}\gtrsim \frac{1}{n^2} \sum_{k=1}^n\frac{1}{k^2} \gtrsim \frac{1}{n^2}.
\end{align}
Moreover, if $z\in [\til_0, \til_1]$, it's easy to see that 
\begin{equation}\label{eq:-3}
\tilde{f}'_{n,L}(z) \geq \frac{\tilde{a}_2}{(\til_2-\til_0)^2}\gtrsim 1.
\end{equation} 
Thanks to \eqref{eq:-1}-\eqref{eq:-3}, we infer that that $\tilde{f}'_{n,L}(z) \asymp \frac{1}{(n+1)^2} \; \text{ for all $\til_n \leq z \leq \til_{n+1}$}$.\\
 Consequently, for $z\in [\til_n, \til_{n+1}]+i \rr$,
 \begin{align}\label{eq:4.3}
|\text{Im} \tilde{f}_{n,L}(z)| \leq |\text{Im}z| \tilde{f}'_{n,L}(\text{Re}z) \lesssim \frac{|\text{Im} z|}{(n+1)^2}.
\end{align}
Finally, consider $z$ which belongs to the square $[\til_n, \til_n+\Delta_n] -i[0, \Delta_n]$ or $[\til_{n+1}-\Delta_n, \til_{n+1}] -i[0, \Delta_n]$. W.o.l.g., assume that $z \in [\til_n, \til_n+\Delta_n] -i[0, \Delta_n]$.\\ 
Then,  there exists $C>0$ such that
\begin{align}
|f_L(z)| &\geq \frac{\tia_n}{|\til_n -z|}-\frac{\tia_{n+1}}{|\til_{n+1} -z|}-|\tilde{f}_{n,L}(z)| \geq \frac{1}{C \Delta_n}-\frac{C}{n}-\beta\frac{\ln n+1}{n}\gtrsim \frac{1}{\Delta_n}
\end{align}
if the constant $\kappa$ in the definition of $\Delta_n$ is chosen to be large.
\end{proof}

\subsection{Large imaginary part}\label{Ss:large part}
For each $n$, another region no containing resonances can be obtained from an estimate on $\text{Im}f_L(z)$. Contrary to the generic case, when $z$ is not too close to the real axis, $|\text{Im}f_L(z)|$ becomes large instead of being small w.r.t. $\left|\frac{1}{L}\text{Im} \left(e^{-i \theta(E)}\right) \right|$. Consequently, there are no resonances. 
\begin{lemma}\label{L:Imf}
We assume the same hypotheses and notations in Lemma \ref{L:modulus} and put $x_0:=L^2(\lambda^{i}_{n+1}-\lambda^{i}_n) \asymp 2n+1$.\\
 Then,  for $1\leq n \leq \ve L/C_{1}$, we have $\left |\text{Im}f_L(z)\right| \gtrsim \frac{1}{\ve L}$ for all $ \frac{x_0^2}{\ve L}\leq  |\text{Im} z| \leq \ve^5L^2$.\\
 Besides, the above statement still holds in the region $[0, \til^{i}_1]-i \left[\frac{1}{\ve L}, \varepsilon^5 L^2\right]$. 
\end{lemma}
\begin{proof}[Proof of Lemma \ref{L:Imf}]
Throughout the proof, we will skip all superscript $i$ in $\til^{i}_{n}, \tia^{i}_{n}$ associated to eigenvalues in $B_{i}$.\\
First of all, we have 
\begin{align}\label{eq:3.35}
\left |\text{Im}f_L(z)\right| \geq \dfrac{\tia_n |\text{Im} z|}{x^2+|\text{Im} z|^2}+\dfrac{\tia_{n+1}|\text{Im} z|}{(x_0-x)^2+|\text{Im} z|^2}+\sum^{\ve L}_{\substack{k=0 \\ k\ne n, n+1}} \dfrac{\tilde{a}_k  |\text{Im} z|}{(\til_k-\text{Re}z)^2+|\text{Im} z|^2}
\end{align}
where $x:=\text{Re}z-\til_n$.\\
Hence,  
\begin{equation}\label{eq:3.41}
\left |\text{Im}f_L(z) \right| \gtrsim \dfrac{ |\text{Im} z|}{x_0^2+|\text{Im} z|^2} \gtrsim \frac{1}{\ve L}\cdot \frac{1}{1+\frac{x_0^2}{\ve^2 L^2}} \gtrsim \frac{1}{\ve L} 
\end{equation}
for all $\frac{x_0^2}{\ve L} \leq |\text{Im}z| \leq \ve L$.\\
Now, assume that $\ve L \leq |\text{Im} z|\leq \varepsilon^5 L^2$, we will find a good lower bound for the last term of RHS of \eqref{eq:3.35}.
We compute
\begin{align}\label{eq:3.39}
A:&= \sum^{\ve L}_{\substack{k=0 \\ k\ne n, n+1}} \dfrac{\tilde{a}_k  |\text{Im} z|}{(\til_k-\text{Re}z)^2+|\text{Im} z|^2}  \notag\\
&= \sum_{k=0}^{n-1} \dfrac{\tia_k |\text{Im} z|}{(\til_k-\text{Re}z)^2+|\text{Im} z|^2} + \sum_{k= n+2}^{\ve L} \dfrac{\tia_k |\text{Im} z|}{(\til_k-\text{Re}z)^2+|\text{Im} z|^2}\\
& \geq \sum_{k= n+2}^{\ve L} \dfrac{\tia_k |\text{Im} z|}{|\text{Im} z|^2+C (k-n)^2(k+n)^2}\notag\\
&\gtrsim \sum_{k= 2}^{\ve L/2}\dfrac{ |\text{Im} z|}{C k^2(k+2n)^2+|\text{Im} z|^2}\gtrsim  y^{1/2} \int_{2}^{\frac{1}{2}\ve L}\dfrac{dt}{C t^2(t+2n)^2+y}\notag
\end{align}
where $y:=|\text{Im} z|^2 \geq \ve^2 L^2 \gg 1 $.\\
Let's assume that $n\geq 1$. By the change of variables $t=y^{1/4}u$, we have 
\begin{align*}
B:=\int_{2}^{\frac{1}{2}\ve L}\dfrac{dt}{t^2(t+2n)^2+y}= y^{-3/4}\int_{2y^{-1/4}}^{\frac{1}{2}\ve L y^{-1/4}}\dfrac{du}{C u^2(u+2ny^{-1/4})^2+1}.
\end{align*}
Note that $\ve L y^{-1/4}= \frac{\ve L}{\sqrt{|\text{Im}z|}} \geq \frac{\ve L}{\ve^2 L}=\frac{1}{\ve}$ for all $|\text{Im} z|\leq \varepsilon^5 L^2$.
Hence, for $2 \ve < 10^{-3}$, we have 
\begin{align}\label{eq:3.40}
A&\gtrsim y^{-1/4}\int_{2}^{1000}\dfrac{du}{u^2(u+2ny^{-1/4})^2+1}\notag\\
&\gtrsim \dfrac{1}{\sqrt{|\text{Im} z|}}\int_{2}^{1000}\dfrac{du}{C u^2(u+2ny^{-1/4})^2+1}.
\end{align}
We observe that, if $\dfrac{n}{y^{1/4}}=\dfrac{n}{\sqrt{|\text{Im} z|}}$ is smaller than a positive constant, say $\alpha$ i.e., $|\text{Im} z| \geq n^2/\alpha$, the above integral is lower bounded by a positive constant $C_{\alpha}$. Then, 
$$A\gtrsim \dfrac{C_{\alpha}}{\sqrt{|\text{Im} z|}}\gtrsim  \frac{1}{\ve^2 L}  \text{ when } \frac{n^2}{\alpha}\leq |Im z|\leq \ve^5 L^2.$$ 
Note that, if $\ve L \geq \frac{n^2}{\alpha}$ i.e., $n \lesssim \sqrt{\ve L}$, the above inequality holds true for all $\ve L \leq |\text{Im}z| \leq \ve^5 L^2$.\\
Finally, we consider the case $n\gtrsim \sqrt{\ve L}$ and find a lower bound for $|\text{Im} f_L(z)|$ in the domain $\ve L\leq  |\text{Im} z| \leq \frac{n^2}{\alpha}$ where $\alpha$ is a large, fixed constant.\\
Thanks to the first inequality in (\ref{eq:3.39}), we have 
\begin{align}\label{eq:3.100}
A& \gtrsim  \sum_{k=0}^{n-1} \dfrac{|\text{Im} z|}{(\til_k-\text{Re}z)^2+|\text{Im} z|^2} \gtrsim \sum_{k=0}^{n-1} \dfrac{|\text{Im} z|}{(\til_{n+1}-\til_k)^2+|\text{Im} z|^2} \notag\\ 
&\gtrsim 
 \sum_{k=0}^{n-1} \dfrac{|\text{Im} z|}{(n+1 -k)^2(n+1+k)^2 +|\text{Im} z|^2}\\
&\gtrsim \sum_{k=2}^{n+1}\dfrac{|\text{Im} z|}{(2n+2-k)^2 k^2+|\text{Im} z|^2} \gtrsim  \sum_{k=2}^{n+1}\dfrac{|\text{Im} z|}{n^2 k^2+|\text{Im} z|^2}\notag\\
&=\dfrac{y_1}{n}\sum_{k=2}^{n+1}\dfrac{1}{k^2+y_1^2}\gtrsim \dfrac{y_1}{n}\int_2^{n+2} \dfrac{dt}{t^2+y_1^2}\notag
\end{align}
where $1 \leq \frac{\ve L}{n}\leq y_1:=\frac{|\text{Im} z|}{n}\leq \frac{n}{\alpha}$. Here, we choose $\alpha \geq 3$. Then, by the change of variables $t:=y_1 u$, we have 
\begin{equation}\label{eq:3.42}
A \gtrsim  \dfrac{1}{n} \int_{2/y_1}^{(n+2)/y_1}\dfrac{du}{Cu^2+1} \gtrsim  \dfrac{1}{n} \int_{2}^{\alpha}\dfrac{du}{u^2+1}\gtrsim \dfrac{1}{n}\gtrsim \dfrac{1}{\varepsilon L}
\end{equation}
 for all $\ve L \leq |\text{Im} z|\leq \frac{n^2}{\alpha}$.\\
Thanks to (\ref{eq:3.40})-(\ref{eq:3.42}), we conclude that $\left |\text{Im}f_L(z)\right|\geq \dfrac{C}{\varepsilon L}$ with $\ve L \leq |\text{Im} z|\leq \varepsilon^5 L^2$ for all $n\geq 1$.\\ 
Now, we consider the case $\text{Re}z\in [0, \til_1]$. For all $1\leq |\text{Im} z|\leq \varepsilon^5 L^2$, we proceed as in \eqref{eq:3.39} and \eqref{eq:3.40} to get
\begin{align}\label{eq:3.52}
&\left |\text{Im}f_L(z)\right|\geq c_0\sum_{k= 2}^{\ve L}\dfrac{|\text{Im} z|}{\alpha^2 k^4+|\text{Im} z|^2}
 \geq c_0 \sum_{k=2}^{\ve L}\dfrac{|\text{Im} z|}{\alpha^2 k^4+|\text{Im} z|^2}\\
 &\geq  \dfrac{c_0}{\sqrt{|\text{Im} z|}}\int_{\frac{2}{\sqrt{|\text{Im} z|}}}^{\frac{\ve L}{\sqrt{|\text{Im} z|}}}\dfrac{du}{\alpha^2 u^4+1}\geq \dfrac{c_0}{\sqrt{|\text{Im} z|}}\int_{2}^{10000}\dfrac{du}{\alpha^2  u^4+1} \gtrsim \dfrac{1}{\varepsilon^2 L}\notag.
\end{align}
On the other hand, for $0<|\text{Im} z|<1$, by putting $t=\frac{1}{\sqrt{|\text{Im} z|}}\geq 1$, we have
\begin{equation}\label{eq:3.50}
\left |\text{Im}f_L(z)\right|\geq t \int_{2t}^{\ve L t}\dfrac{du}{Cu^4+1}\geq \dfrac{t}{C} \int_{2t}^{\ve L t}\dfrac{du}{u^4}\gtrsim \frac{1}{t^2} \gtrsim  |\text{Im} z|.
\end{equation}
Thanks to \eqref{eq:3.52} and \eqref{eq:3.50}, $\left |\text{Im}f_L(z)\right|\gtrsim \dfrac{1}{\ve L}$ for all $\frac{1}{\varepsilon L}\leq |\text{Im} z| \leq \ve^5 L^2$ and $\text{Re}z\in [0, \til_1]$. Hence, the claim follows.
\end{proof}
Thanks to Lemmata \ref{L:modulus} and \ref{L:Imf}, we obtain free resonance regions illustrated in Figures \ref{F:f1}-\ref{F:f4}.
\section[Resonances closest to the real axis]{Resonances closest to the real axis}\label{Ss:closest}
In the present section, for each band $B_{i}$ of $\Sigma_{\zz}$, we will study rescaled resonances in $\Omega^{i}_n, \tilde{\Omega}^{i}_n$ and $\Omega^{i}$ (see Figures \ref{F:f1}-\ref{F:f4}).\\
 \textbf{Convention:} Recall that we use the (local) enumeration $(\ld^{i}_{\ell})_{\ell \geq 0}$ for (distinct) eigenvalues in the band $B_{i} \ni E_{0}$ and the usual enumeration $(\ld_{k})$ for eigenvalues of $H_{L}$ outside the band $B_{i}$ (written in increasing order and repeated according to their multiplicity).
 In the proofs of all results stated in this section, we will suppress the superscript $i$ in $\til^{i}_{k}, \tia^{i}_{k}$ in order to lighten the notation. We will only specify the superscript $i$ in case there is a risk of confusion. Note that, whenever we refer to $\ld_{n}, \ld_{n+1}$ in this section, they are respectively $\ld^{i}_{n}, \ld^{i}_{n+1}$, the $(n+1)-$th and $(n+2)-$th eigenvalues in the band $B_{i}$. However, we will always use the notations $\ld_{k}$ or $\til_{k}$ to refer to the eigenvalues with the usual enumeration which does not depends on bands of $\Sigma_{\zz}$. Finally, as an abuse of notations, $\sum\limits_{k\ne n}$ and $\sum\limits_{k \ne n, n+1}$ stand for, respectively, $\sum \limits_{\ld_{k} \ne \ld^{i}_{n}}$ and $\sum\limits_{\ld_{k} \ne \ld^{i}_{n}, \ld^{i}_{n+1}}$. \\
 When $E_{0}=\inf \Sigma_{\zz}$ and we ignore eigenvalues outsider $\Sigma_{\zz}$, two enumerations will be the same and readers can actually think of this case while following our proof. 
\subsection{Resonances in $\Omega^i_n$}\label{Ss:nlarge}
Recall that the region $\Omega^i_n$ corresponds to the case $\Delta_n < \frac{x_0^2}{\ve L}$ with $x_0=\til^{i}_{n+1}-\til^{i}_n$ which is equivalent to $\kappa (n+1)(\ln (n+1) +1) \gtrsim \ve L$ (see Lemma \ref{L:modulus} for the def. of $\Delta_n$). Then, $n\geq \frac{\eta L}{\ln L}$ with some small $\eta \asymp \frac{\ve}{\kappa}$.\\ 
The schema of studying resonances in $\Omega^i_n$ is split into two steps:\\ 
In Step $1$, we will show that the number of solution of the resonance equation \eqref{eq:4.0} is equal to that of the following equation by using Rouch\'{e}'s theorem. 
\begin{equation}\label{eq:4.4}
f(z):=f_L(z)+\frac{1}{L}e^{-i \theta(E_0)}=0
\end{equation}
Hence, we reduce our problem to count the number of solutions of \eqref{eq:4.4}. Note that, this number is exactly the cardinality of the set $f^{-1}_L \left( \left \{ -\frac{1}{L}e^{-i \theta(E_0)} \right \} \right)$, the inverse image of the number $-\frac{1}{L}e^{-i \theta(E_0)}$.\\
Next, in Step $2$, we partition $\Omega^i_n$ into two parts, the rectangles $ABCD$ and $EFGH$.
First of all, we will show that the image of  the boundary of the rectangle $ABCD$ under $f_L$ is still a simple contour and on this contour, $|f'_L(z)| \gtrsim \frac{1}{n^2}$. Then, by the Argument Principal  to the holomorphic function $f_L$ in $\Omega^i_n$, we infer that $f_L$ is a conformal map from $ABCD$ onto $f_L(ABCD)$ and its inverse is holomorphic as well. Hence, there is at most one resonance in this domain. The existence of that unique resonance depends on whether $f_L(ABCD)$ contains the point $-\frac{1}{L}e^{-i \theta(E_0)}$ or not.\\
Moreover, by studying $f_{L}(EFGH)$, we can conclude that there is at least one resonance which stays either in $ABCD$ or $EFGH$. Besides, if there is a resonance in $ABCD$, that will be the unique resonance in $\Omega^i_n$.\\
Let's start the present subsection with the proof of the statement in Step $1$:
\begin{lemma}\label{L:compare}
The equations \eqref{eq:4.0} and \eqref{eq:4.4} have the same number of solutions in $\Omega^i_n$.
\end{lemma}
\begin{proof}[Proof of Lemma \ref{L:compare}]
Define $f(z)$ as in \eqref{eq:4.4} and $g(z):=f_L(z)+\frac{1}{L}e^{-i \theta (E)}$.\\
First of all, we observe that  $f$ and $g$ are holomorphic in $\overline{\Omega^{i}_{n}}$ since $\theta(E)$ is holomorphic in $[E_0, E_0+\ve^2] +i\left[-\frac{x_0^2}{\ve L^3}, 0\right]$ for all $E_0 \in (-2, 2)$.\\ 
Moreover,
\begin{align*}\label{eq:4.5}
|f(z)-g(z)|&=\frac{1}{L} \left|e^{-i\theta(E)}-e^{-i\theta(E_0)}\right| \notag\\ 
&\leq \frac{1}{L}\left |e^{-i\theta(E)}-e^{-i\theta(\text{Re}E)}\right| +\frac{1}{L}\left |e^{-i\theta(\text{Re}E)}-e^{-i\theta(E_0)}\right|\\
& \leq \frac{C}{L} |\text{Im}E | +\frac{C}{L}|\text{Re}E-E_0| \leq \frac{C}{L}\cdot \frac{n^2}{\ve L^3} +\frac{C \ve^2}{L}\leq \frac{C \ve^2}{L}\notag
\end{align*}
where the constant $C$ is independent of $\ve$.\\
Hence, to carry out the proof of the present lemma, it suffices to show that 
\begin{equation*}\label{eq:4.-1}
|f(z)| \gtrsim \frac{1}{L} \text{ on the contour $\gamma_n=\partial \Omega^{i}_n$}.
\end{equation*}
Indeed, assume that we have such an estimate for $f(z)$ on $\gamma_n$. Then, $|f(z)-g(z)| \leq |f(z)|$ on $\gamma_n$. Hence, thanks to Rouch\'{e}'s theorem, $f$ and $g$ have the same number of zeros in the domain $\Omega^i_n$.\\
Moreover, observe that $f_L(z)$ is real iff $z$ is real. Hence, for $z\in \rr$, 
$$|f(z)| \geq |\text{Im} f(z)| =\frac{1}{L}\left | \text{Im}\left(e^{-i\theta(E_0)}\right) \right |= \frac{|\sin \theta(E_0)|}{L}.$$
Note that, $\sin(\theta(E_0)) \ne 0$ since $E_0\in (-2, 2)$.\\     
Hence, to prove \eqref{eq:4.-1}, it is sufficient to show that 
\begin{equation}\label{eq:4.6}
|f_L(z)| \gtrsim \frac{1}{\ve L} \text{ on $\gamma_n \backslash \rr$}.
\end{equation}
We decompose the contour $\gamma_n$ into horizontal and vertical line segments as in Figure \ref{F:f1} .\\
First of all, on the segments $AD, BC, DE, CH$, $|f_L(z)|$ is big (the zone near poles $\til_n$ and $\til_{n+1}$ of $f_L(z)$). More precisely, according to Lemma \ref{L:modulus}, on these segments, 
\begin{equation}\label{eq:4.7}
|f_L(z)| \gtrsim \frac{1}{\Delta_n}\gtrsim \kappa \frac{\ln (n+1)+1}{n+1} \gtrsim \frac{1}{\ve L}.
\end{equation}
Next, on the segment $FG$, by Lemma \ref{L:Imf}, we have 
\begin{equation}\label{eq:4.8}
|\text{Im} f_L(z)| \gtrsim \frac{1}{\ve L}.
\end{equation}
Finally, we study $f_L(z)$ on $EF$ and $GH$. It suffices to consider the segment $EF$ as $\til_n$ and $\til_{n+1}$ play equivalent roles.\\
Let $z\in EF$, hence, $z=\til_n- it$ with $\Delta_n \leq t \leq \frac{x_0^2}{\ve L}$.\\ 
Then, 
\begin{equation*}\label{eq:4.10}
|\text{Im}f_L(z)| \geq |\text{Im}f_{n, L}(z)|=\frac{\tia_n}{t} + \frac{\tia_{n+1} t}{(\til_{n+1}-\til_n)^2+t^2 } \gtrsim \varphi(t)
\end{equation*}
where $\varphi(t):=\frac{1}{t} +\frac{t}{(\til_{n+1}-\til_n)^2+t^2}$.\\
It's easy to check that $\varphi'(t)= -\frac{1}{t^2} +\frac{(\til_{n+1}-\til_n)^2-t^2}{\left[(\til_{n+1}-\til_n)^2+t^2\right]^2} \leq 0 $ for all $t\ne 0$. Hence, $\varphi(t)$ is (strictly) decreasing in the interval $\left[\Delta_n, \frac{x_0^2}{\ve L}\right]$. Therefore,
\begin{equation}\label{eq:4.11}
|\text{Im}f_L(z)| \gtrsim \varphi\left(\frac{x_0^2}{\ve L}\right) \gtrsim \frac{\ve L}{n^2} + \frac{1}{\ve L \left(1+\frac{n^2}{(\ve L)^2} \right)} \gtrsim \frac{1}{\ve L}.
\end{equation}
Thanks to \eqref{eq:4.7}-\eqref{eq:4.11}, the claim in \eqref{eq:4.6} follows and we have Lemma \ref{L:compare} proved. 
\end{proof}
Now, we describe the image of the rectangles $ABCD$ and $EFGH$.  First of all, we consider the rectangle $ABCD$ which is closer to the real axis. 
\begin{lemma}\label{L:simplecontour}
Let $ABCD$ be the rectangle $[\til^{i}_n+\Delta_n, \til^{i}_{n+1}-\Delta_n] +i[-\Delta_n, 0]$
and $\gamma^1_n$ be its boundary.\\ 
Then, $f_L(\gamma^1_n)$ is a simple contour. Besides, we have $|f_L'(z)| \gtrsim \frac{1}{n^2}$ on $\gamma^1_n$.
\end{lemma}
\begin{proof}[Proof of Lemma \ref{L:simplecontour}]
First of all, on the horizontal segment $AB$ where $\til_n+\Delta_n \leq z \leq \til_{n+1}-\Delta_n$, $f_L$ is real-valued and 
$$f_L'(z)=\sum_k\frac{\tia_k}{(\til_k-z)^2} \gtrsim \frac{1}{(\til_{n+1}-\til_n)^2} \gtrsim \frac{1}{n^2} >0.$$
Hence, $f_L(z)$ is strictly increasing on $AB$. Then, $f_L$ is injective and it transforms $AB$ into an interval $[m^1_{-}, m^1_{+}]$ in $\rr$. Note that, since Lemma \ref{L:modulus}, we have  
\begin{align}
&m^1_{-}=f_L(\til_n+\Delta_n)\\
&= -\frac{\tia_n}{\Delta_n}+\frac{\tia_{n+1}}{\til_{n+1}-\til_n-\Delta_n}+\tilde{f}_{n,L}(\til_n+\Delta_n) \lesssim -\frac{1}{\Delta_n}<0\notag
\end{align}
if the constant $\kappa$ in the definition of $\Delta_n$ is large enough.\\
Similarly, $m^1_{+}=f_L(\til_{n+1}-\Delta_n) \gtrsim \frac{1}{\Delta_n} \gtrsim \frac{\ln L}{L}$ for all $n>\frac{\eta L}{\ln L}$.\\
Now, for $z=x+iy\in \mathbb C$, we have 
\begin{align}\label{eq:4.12}
f'_L(z)&=\sum_k \frac{\tia_k}{(\til_k-z)^2}= \sum_k \frac{\tia_k}{(\til_k-x -iy)^2}\\
&=\sum_k \frac{\tia_k}{(\til_k-x)^2}\cdot \frac{1}{\left (1-\frac{i y}{\til_k-x}\right)^2}=\sum_k \frac{\tia_k}{(\til_k-x)^2}\cdot
\frac{\left (1+\frac{i y}{\til_k-x}\right)^2}{ \left[ 1+\left(\frac{y}{\til_k-x} \right)^2 \right]^2}
\notag
\end{align}
Note that, for any holomorphic function $f$, $f'(z)=\frac{\partial f}{\partial x}=\frac{1}{i}\frac{\partial f}{\partial y} $ with $z=x+i y$. Hence, $\frac{\partial}{\partial x}\text{Re} f_L(z)= \frac{\partial}{\partial y}\text{Im} f_L(z)= \text{Re}[f_L'(z)]$ and we put 
\begin{align}\label{eq:4.13}
p(x,y):= \text{Re}[f_L'(z)]= \sum_k  \frac{\tia_k}{(\til_k-x)^2} \cdot \underbrace{\frac{1-\left( \frac{y}{\til_k-x}  \right)^2}{ \left[ 1+\left(\frac{y}{\til_k-x} \right)^2 \right]^2}}_{=:p_k(x,y)}.
\end{align}
Next, we study $f_L$ on the line segment $AD$ where $z=x+ iy$ with $x=\til_n+\Delta_n$ and $-\Delta_n \leq y \leq 0$. In the present case, the identity \eqref{eq:4.13} reads
\begin{align}\label{eq:4.27}
\text{Re}[f_L'(z)]&=\tia_n \frac{\Delta_n^2-y^2}{(\Delta_n^2+y^2)^2}+
\sum_{k\ne n}  \frac{\tia_k}{(\til_k-x)^2}\cdot \frac{1-\left( \frac{y}{\til_k-x}  \right)^2}{ \left[ 1+\left(\frac{y}{\til_k-x} \right)^2 \right]^2}.
\end{align}
along the segment $AD$.\\
Note that, for any $\ld_{k} \ne \ld^{i}_{n}$, we have $(\til_k-x)^2 \gtrsim n^2\gg \Delta_n^2$ for all $z\in AD$. Recall that we are considering the case that $n \geq \eta \frac{L}{\ln L}$.\\
Hence, all the terms in RHS of \eqref{eq:4.27} are positive for all $y \in [-\Delta_n, 0]$. This implies that
$$\frac{\partial}{\partial y}\text{Im}f_L(z) =\text{Re}[f_L'(z)] \gtrsim \sum_{k \ne n} \frac{\tia_k}{(\til_k-x)^2}\gtrsim \frac{1}{n^2} .$$
Hence, along the segment $AD$, $\text{Im}f_L(z)$ is strictly increasing in $y=\text{Im}z$. As a result, $f_L(z)$ is injective on $AD$ and 
\begin{align}\label{eq:4.32}
&0 \geq \text{Im}f_L(z) \geq \text{Im}f_L(D) =  \text{Im}f_L(\til_n+\Delta_n -i \Delta_n)\\
&=-\frac{\tia_n}{2 \Delta_n} -\sum_{k\ne n} \frac{\tia_k \Delta_n}{(\til_k-\til_n)^2+\Delta_n^2} \asymp -\frac{1}{\Delta_n}-\frac{\Delta_n}{n^2} \asymp -\frac{1}{\Delta_n}\notag.
\end{align}
Note that, on $AD$, the $n-$th term $\frac{\tia_n y}{y^2+\Delta_n^2}$ is much bigger than the sum of the other terms in $\text{Im}f_L(z)$. Precisely,
\begin{equation*}\label{eq:4.47}
\text{Im}f_L(z)= \frac{\tia_n y}{y^2+\Delta_n^2} \left (1+ O\left( \frac{1}{\kappa^2 \ln^2 n}\right) \right)=\frac{\tia_n y}{y^2+\Delta_n^2} \left (1+ O\left( \frac{1}{\kappa^2 \ln^2 L}\right) \right).
\end{equation*}
Hence, the monotonicity of $\text{Im}f_L(z)$ in $y$ on $AD$ just comes from that of $\frac{\tia_n y}{y^2+\Delta_n^2}$.\\
Next, we will estimate $\text{Re}f_L(z)$ on $AD$. For all $-\Delta_n \leq y \leq 0$, by Lemma \ref{L:modulus}, we have 
\begin{align}\label{eq:4.33}
\text{Re}f_L(z)&= -\frac{\tia_n \Delta_n}{\Delta^2_n+y^2}+ \frac{\tia_k (\til_{n+1}-\til_n-\Delta_n)}{(\til_{n+1}-\til_n-\Delta_n)^2+y^2}+\text{Re}\tilde{f}_{n,L}(z) \notag\\ 
&\asymp -\frac{1}{\Delta_n}+O\left( \frac{\ln n}{n}\right).
\end{align}
Hence, when we choose the constant $\kappa$ in the definition of $\Delta_n$ to be big enough, $-\frac{\tia_n \Delta_n}{\Delta^2_n+y^2}$ becomes the dominating term in RHS of \eqref{eq:4.33}. Then, $\text{Re}f_L(z)\asymp -\frac{1}{\Delta_n}$.\\
By the equivalent role between $\til_n$ and $\til_{n+1}$, we obtain a similar result for the image of $BC$ under $f_L$, that is, 
the $\text{Im}f_L(z)$ is increasing in $y=\text{Im}z \in [-\Delta_n, 0]$, $\text{Re}f_L(z) \asymp \frac{1}{\Delta_n}$ and $|f_L'(z)| \gtrsim \frac{1}{n^2}$.\\  
Finally, we consider $f_L$ on 
$CD=\{z=x-i\Delta_n | x\in[\til_n+\Delta_n, \til_{n+1}-\Delta_n]\}$.\\ 
Note that, in the present case, the function $p_k(x, y)$ in \eqref{eq:4.13} is strictly positive for any $k\ne n, n+1$ and non-negative otherwise. Hence, $\text{Re} f_L(z)$ is strictly increasing in $x$. 
Hence, on $CD$,
\begin{equation*}\label{eq:4.48}
 -\frac{1}{\Delta_n} \asymp \text{Re} f_L(D) \leq \text{Re} f_L(z) \leq \text{Re} f_L(C) \asymp \frac{1}{\Delta_n}. 
\end{equation*}
Moreover, we have the following estimate for $|f_L'(z)|$: 
\begin{align*}\label{eq:4.49}
|f_L'(z)|\geq \text{Re}[f_L'(z)] \gtrsim \frac{1}{n^2}\cdot \frac{1-\left( \frac{\Delta_n}{\til_{n+2}-x}  \right)^2}{ \left[ 1+\left(\frac{\Delta_n}{\til_{n+2}-x} \right)^2 \right]^2} \gtrsim \frac{1}{n^2}
\end{align*}
for all $x\in[\til_n+\Delta_n, \til_{n+1}-\Delta_n]$.\\
Finally, we give estimates on $\text{Im}f_L(z)$ on $CD$. First of all, on this segment, we have
\begin{equation*}\label{eq:4.50}
-\text{Im}f_{n,L}(z)=\frac{\tia_n \Delta_n}{(\til_n-x)^2+\Delta_n^2}+\frac{\tia_{n+1} \Delta_n}{(\til_{n+1}-x)^2+\Delta_n^2}.
\end{equation*} 
It's easy to see that, as $x$ varies in $[\til_n+\Delta_n, \til_{n+1}-\Delta_n]$, $\frac{\Delta_n}{n^2} \lesssim -\text{Im}f_{n,L}(z) \lesssim \frac{1}{\Delta_n}.$\\ 
On the other hand, for $x\in[\til_n+\Delta_n, \til_{n+1}-\Delta_n]$,
\begin{align*}
-\text{Im} \tilde{f}_{n, L}(z)=\Delta_{n }\sum_{k\ne n, n+1} \frac{\tia_{k}}{(\til_{k}-x)^2+\Delta_{n}^{2}} \asymp \Delta_{n }\sum_{k\ne n, n+1} \frac{\tia_{k}}{(\til_{k}-x)^2} \asymp \frac{\Delta_{n}}{n^{2}}.
\end{align*} 
Note that, $\frac{1}{\Delta_{n}} \gg \frac{\Delta_{n}}{n^{2}}$. Hence, $\text{Im}f_{L}(z)$ varies from $-\frac{1}{\Delta_{n}}$ to $-\frac{\Delta_{n}}{n^{2}}$ on the $CD$.\\
To sum up, the holomorphic function $f_L$ is injective on each edge of the rectangle $ABCD$. 
Hence, the image of each edge under $f_L$ is a non self-intersecting continuous curve. Obviously, since $f_L(AB)$ is a segment in the real axis, it does not intersect the other curves. It's easy to see that $f_L(AD) \cap f_L(BC)=\emptyset$ as well. However, it's not so evident that $f_L(AD)$ and $f_L(CD)$ only intersect at $f_{L}(D)$. In order to prove that, it is necessary to use the estimate on the derivative of $f_L(z)$.  Note that $|f'_L(z)| \gtrsim \frac{1}{n^2}$ on all edges of $ABCD$. On the other hand, $f_L$ is holomorphic in a neighborhood of the rectangle $ABCD$. Hence, $f_L$ is locally bi-holomorphic near any point on $AB$, $BC$, $CD$, $DA$. Hence, $f_L(AD)$ only intersects $f_L(CD)$ at $f_{L}(D)$. Hence, $f_L(\gamma^1_n)$ is a simple contour and $|f'_L(z)| \gtrsim \frac{1}{n^2}$ on $\gamma^1_n$.
\end{proof}

\begin{lemma}\label{L:bi}Let $ABCD$ be the rectangle $[\til^{i}_n+\Delta_n, \til^{i}_{n+1}-\Delta_n] +i[-\Delta_n, 0]$.\\ 
Then, the function $f_L(z)$ is a bijection from the rectangle $ABCD$ onto $f_L(ABCD)$  and its inverse in $A'B'C'D'=f_L(ABCD)$ is holomorphic as well. Moreover, 
$$|f'_L(z)| \gtrsim \frac{1}{n^2} \text{ for all $z \text{ belonging to the rectangle } ABCD$}.$$
Consequently, $f_L(z)$ is a conformal map in the interior of the rectangle $ABCD$, hence, the angles between boundary curves of $A'B'C'D'$ are all $90 \degree$.   
\end{lemma}
\begin{proof}[Proof of Lemma \ref{L:bi}]
Denote by $\gamma^1_n$ the boundary of the rectangle $ABCD$.
According to Lemma \ref{L:simplecontour}, $f_L(\gamma^1_n)$ is a simple contour and $f_L$ is injective in $\gamma^1_n$. Hence, for any interior point $w$ of $f_L(ABCD)$, the contour $f_L(\gamma^1_n)$ travels counterclockwise around $w$ exactly one times. Hence, 
thanks to Argument Principle, the equation $f_L(z)=w$ in $ABCD$ has the unique solution. In other words, $f_L$ is injective in the interior of the rectangle $ABCD$. By using Open Mapping Theorem, we infer that $f_L(z)$ is bijective from the rectangle $ABCD$ onto $f_L(ABCD)$ and its inverse in $f_L(ABCD)$ is holomorphic. Moreover, $f_L'(z) \ne 0$ for all $z$ in the rectangle $ABCD$. Hence, by using the Maximum Modulus Principle for the holomorphic function $\frac{1}{f'_L}$, we have $ |f'_L(z)| \gtrsim \frac{1}{n^2} \text{ in $ABCD$}.$\\
The holomorphic function $f_L$ is therefore a conformal map in the rectangle $ABCD$ and the claim follows.
\end{proof}
We observe that the domain $MNOP \backslash A'B'C'D'$ is included in the image of  $EFGH$ under $f_L$.
\begin{lemma}\label{L:belong}
Put $MNOP=\left[-\frac{1}{C\Delta_n}, \frac{1}{C\Delta_n}\right] -i\left[0, \frac{2 |\sin\theta(E_0)|}{L}\right]$ with $C>0$ large.
Let $A'B'C'D'= f_{L}(ABCD)$.\\
 Assume that $-\frac{1}{L}e^{-i \theta(E_0)} \in MNOP \backslash A'B'C'D'$. Then, 
$$f_L(EFGH) \supset MNOP \backslash A'B'C'D'.$$
\end{lemma}
We will skip the proof of Lemma \ref{L:belong} for a while and make use of this lemma to describe resonances in the domain $\Omega^i_n$.
\begin{proof}[Proof of Theorem \ref{T:abcd}]
Recall that, thanks to Lemma \ref{L:compare}, the number of rescaled resonances $z$ is the cardinality of $f^{-1}_L \left( \left\{ -\frac{1}{L}e^{-i \theta(E_0)}\right \} \right)$. Hence, for the existence of resonances, we have to check if the point $-\frac{1}{L}e^{-i \theta(E_0)}$ belongs to $\Omega^{i}_n$. Note that $-\frac{1}{L}e^{-i \theta(E_0)}$ always stays inside the open rectangle $MNOP=\left[-\frac{1}{C\Delta_n}, \frac{1}{C\Delta_n}\right] -i\left[0, \frac{2 |\sin \theta(E_0)|}{ L}\right]$ where $C>0$ is a big constant. We consider two possibilities. First of all, assume that $-\frac{1}{L}e^{-i \theta(E_0)}$ belongs to $A'B'C'D'$. Then, by Lemma \ref{L:bi}, there exists one and only one rescaled resonance $z_n$ in $ABCD$. When that case happens, $|\text{Im}z_n| \leq \Delta_n=\frac{n}{\kappa \ln n}\asymp \frac{n}{\kappa \ln L}$. Remark that, this case can not happen for all $\ve L \geq n > \eta \frac{L}{\ln L}$. For example, when $n=\ve L$ i.e., the real part of resonance is far from $\partial \Sigma_{\mathbb Z}$ by a constant distance, there are no rescaled resonances in $ABCD$ according to \cite[Theorem 1.2]{klopp13}.\\
Now, assume that the other case happens i.e., $-\frac{1}{L}e^{-i \theta(E_0)} \in MNOP \backslash A'B'C'D'.$\\
In this case, $f_L(EFGH)$ contains $MNOP \backslash A'B'C'D'$ by Lemma \ref{L:belong}. Then, $-\frac{1}{L}e^{-i \theta(E_0)}$ stays in the image of $EFGH$ under $f_L$. Hence, there exists a rescaled resonance in the rectangle $EFGH$ and note that the imaginary part of such a rescaled resonance is smaller than $-\frac{x_0^2}{\ve L}$ and bigger than $-\Delta_n$.   
\end{proof}
Finally, to complete the subsection, we state here the proof of Lemma \ref{L:belong}.
\begin{proof}[Proof of Lemma \ref{L:belong}]
Note that $|\text{Re}f_L(z)|$ is bigger than $\frac{1}{C\Delta_n}$ on segments $AD$ and $BC$ if $C$ is large enough.\\ 
Then, the hypothesis that $-\frac{1}{L}e^{-i \theta(E_0)} \in MNOP \backslash A'B'C'D'$ yields $\frac{\Delta_n}{n^2}\leq 2\frac{|\sin(\theta(E_0))|}{L}$. Hence,
$$ \frac{\Delta_{n}}{n^{2}}\leq 2\frac{|\sin(\theta(E_0))|}{L} \leq \frac{1}{\ve L} \ll \frac{1}{\Delta_{n}}.$$
By the open mapping theorem, the image of the open rectangle $EFGH$ is still a bounded domain in $\mathbb C$. From the study of the curve $f_L(CD)$ in Lemma \ref{L:simplecontour}, we know that, the imaginary part of $f_L(CD)$ increases from $-\frac{1}{\Delta_n}$ to $-\frac{\Delta_n}{n^2}$.  
Hence, it suffices to show that the imaginary part of $f_L$ on all parts of the boundary of $EFGH$ except for $CD$ is smaller than $-\frac{1}{\ve L}$ up to a constant factor.\\
First of all, by Lemma \ref{L:Imf},
$\text{Im}f_L(z) \lesssim -\frac{1}{\ve L}$ on $FG$. 
Next, we consider segments $ED$ and $CH$. By symmetry, it suffices to study the image of $f_L$ on $ED$.\\
Let $z\in ED$. Then, $z=x-i\Delta_n$ with $x\in [\til_n, \til_n+\Delta_n]$,  
\begin{align}\label{eq:4.34}
-\text{Im}f_L(z)=\frac{\tia_n \Delta_n}{(\til_n-x)^2+\Delta_n^2} + \Delta_n\sum_{k\ne n}\frac{\tia_k }{(\til_k-x)^2+\Delta_n^2}.
\end{align}
According to Lemma \ref{L:modulus}, the second term of RHS of \eqref{eq:4.34} is bounded by $\frac{\Delta_n}{n^2}$. On the other hand, the first term is bigger than $\frac{\tia_n}{2 \Delta_n}\gg \frac{\Delta_n}{n^2}$ since $n>\frac{\eta L}{\ln L}$.
Hence, 
\begin{equation}\label{eq:4.44}
-\text{Im}f_L(z)=\frac{\tia_n \Delta_n}{(\til_n-x)^2+\Delta_n^2} \left (1+O\left(\frac{1}{\ln^2 L} \right) \right)
\end{equation}
uniformly in $x\in [\til_n, \til_n+\Delta_n]$.
Since the function $\frac{\tia_n \Delta_n}{(\til_n-x)^2+\Delta_n^2}$ is decreasing in $x\in [\til_n, \til_n+\Delta_n]$, we infer that $\text{Im}f_L(z)$ is strictly increasing, hence, $f_L(z)$ is injective on $ED$. Moreover, 
$$\text{Im}f_L(z)\asymp -\frac{1}{\Delta_n} \ll -\frac{1}{\ve L}.$$
We consider now $f_L(z)$ on the vertical segment $EF$ where $z=x+i y$ with $x\equiv \til_n$ and $-\frac{x_0^2}{\ve L} \leq y \leq -\Delta_n$.
Then,
\begin{equation}\label{eq:4.23}
f_L'(z)=-\frac{\tia_n}{y^2} + \sum_{k\ne n} \frac{\tia_k}{(\til_k-\til_n)^2}\cdot
\frac{\left (1+\frac{i y}{\til_k-\til_n}\right)^2}{ \left[ 1+\left(\frac{y}{\til_k-\til_n} \right)^2 \right]^2}.
\end{equation}
Hence, 
\begin{equation}\label{eq:4.24}
-\frac{\partial}{\partial y}\text{Im}f_L(z) =-\text{Re}[f_L'(z)]= \frac{\tia_n}{y^2}-\underbrace{\sum_{k \ne n} \frac{\tia_k}{(\til_k-\til_n)^2}\cdot \frac{1-\left( \frac{y}{\til_k-\til_n}  \right)^2}{ \left[ 1+\left(\frac{y}{\til_k-\til_n} \right)^2 \right]^2}}_{=:s(y)}
\end{equation}
 For $k\ne n$, let $u_k:=\frac{y^2}{(\til_k-\til_n)^2}$ and $\psi(u_k):=\frac{1-u_k}{(1+u_k)^2}$. Then, 
$\psi'(u_k)=\frac{u_k-3}{(u_k+1)^3}$.\\
Note that, in the present case,
$$0 < u_k \lesssim \frac{\left(\frac{n}{\ve L} \right)^2 \cdot n^2}{n^2} \lesssim \left(\frac{n}{\ve L} \right)^2. $$
Hence, for any $n \leq \frac{\ve L}{C}$ with $C$ large and $k \ne n$, $u_k \in (0, 1/2]$.\\
Hence, $\psi(u_k)$ is decreasing and $\frac{2}{9}=\psi\left(\frac{1}{2}\right) \leq \psi(u_k) \leq \psi(0)=1.$\\
Therefore, there exists a numeric constant $\mu$ s.t. 
\begin{equation}\label{eq:4.25}
\frac{1}{\mu n^2} \leq s(y) \leq \frac{\mu}{n^2}.
\end{equation}
\eqref{eq:4.24} and \eqref{eq:4.25} yield that 
\begin{equation}\label{eq:4.26}
-\frac{\partial}{\partial y}\text{Im}f_L(z) \geq \frac{c_0 (\ve L)^2}{n^4} -\frac{\mu}{n^2} \geq \frac{\mu}{n^2}
\end{equation}
for all $n  \leq \frac{\ve L}{C_1}$ with $C_1=C_1(\alpha, c_0, \mu)$ large enough.\\
Hence, in the present case, $\text{Im}f_L(z)$ is decreasing in $y$. As a result, the function $f_L(z)$ is injective on $EF$ and 
$|f'_L(z)| \geq |\text{Re}[f'(z)]| \gtrsim \frac{1}{n^2}$.\\
Besides, on $EF$,
\begin{align}\label{eq:4.30}
\text{Im}f_L(z) &\geq \text{Im}f_L(\til_n-i\Delta_n)=\text{Im}f_L(E) \\
&=-\frac{\tia_n}{\Delta_n}-\Delta_n \sum_{k\ne n} \frac{\tia_k}{(\til_k-\til_n)^2+\Delta_n^2} \asymp -\frac{1}{\Delta_n};\notag
\end{align}
and 
\begin{align}\label{eq:4.31}
&\text{Im}f_L(z) \leq \text{Im}f_L\left(\til_n-i \frac{x_0^2}{\ve L}\right)= \text{Im}f_L(F) \\
&\asymp -\frac{\tia_n \ve L}{n^2}-\frac{n^2}{\ve L} \sum_{k\ne n} \frac{\tia_k}{(\til_k-\til_n)^2+\frac{n^4}{(\ve L)^2}}\asymp -\frac{\ve L}{n^2} \lesssim -\frac{1}{\ve L}\notag.
\end{align}
By symmetry, we have the same conclusion for the image of $GH$ under $f_L$.\\
To sum up, the images of $ED, CH, EF, FG, GH$ under $f_L$ stay below the horizontal $y=-\frac{\tilde{C}}{\ve L}$ in the complex plane with some positive constant $\tilde{C}$. Hence, the claim follows. 
\end{proof}
\subsection{Resonances in $\tilde{\Omega}^{i}_n$ and $\Omega^{i}$}\label{Ss:nsmall}
First of all, all sides of the rectangle $\tilde{\Omega}^i_n$ are included in horizontal and vertical segments of $\Omega^i_n$. Hence, Lemma \ref{L:compare} still hold for $\tilde{\Omega}^i_n$. We will prove the existence and uniqueness of rescaled resonances in $\tilde{\Omega}^i_n$. 
\begin{proof}[Proof of Theorem \ref{T:contour1}]
Let $\tilde{\gamma}_n$ be the boundary of $\tilde{\Omega}^i_n$.\\
It's easy to check that, the monotonicity and the estimates we made for the real and imaginary part of $f_L(z)$ and $|f'_L(z)|$ on $AB$, $BC$, $AD$ in Lemma \ref{L:simplecontour} still hold for $A_1B_1$, $A_1 D_1$, $B_1 C_1$ of the contour $\tilde{\gamma}_n$.  
Now, we study the image of $C_1 D_1$ under $f_L$. Let $z=x+iy \in C_1 D_1$ with $x\in [\til_n+\Delta_n, \til_{n+1}-\Delta_n]$ and $y \equiv -\frac{x_0^2}{\ve L}$.\\
Note that, in the present case, $\Delta_n \geq \frac{x_0^2}{\ve L}$. Hence,$|\til_k-x| \geq |y|$ for all $k$.
Moreover, for $k \ne n, n+1$, $|\til_k-x| \gtrsim n \gg |y|$. Then, since \eqref{eq:4.13}, we have  
\begin{equation}\label{eq:}
\text{Re}[f'_L(z)] \gtrsim \sum_{k \ne n, n+1} \frac{1}{(\til_k -x)^2} \gtrsim \frac{1}{n^2}
\end{equation}
for all $z \in C_1D_1$.\\
Hence, $\text{Re}f_L(z)$ is still strictly increasing in $x$ on $C_1D_1$. Finally, we compute the magnitude of $\text{Im} f_L(z)$ on $C_1D_1$.
\begin{align}\label{eq:4.51}
-\text{Im}f_L(z) &\geq -\text{Im}f_{n,L}(z) \asymp |y| \left( \frac{1}{(\til_n-x)^2+y^2}+\frac{1}{(\til_{n+1}-x)^2+y^2}\right) \notag\\
& \gtrsim \frac{|y|}{n^2+y^2} \gtrsim \frac{1}{\ve L}.
\end{align} 
Hence, $\text{Im}f_L(z) \lesssim -\frac{1}{\ve L}$. Then, using the same argument as in Lemmata \ref{L:simplecontour} and \ref{L:bi}, we infer that $f_L$ is bijective from $\tilde{\Omega}^i_n$ on $f_L(\tilde{\Omega}^i_n)$. Moreover, thanks to \eqref{eq:4.51}, we deduce that the point $-\frac{e^{-i \theta(E_0)}}{L}$ belongs to $f_L(\tilde{\Omega}^i_n)$. Hence, there exists a unique rescaled resonance in $\tilde{\Omega}^i_n$.  
\end{proof}
Finally, we show that, there are no rescaled resonances in $\mathcal R^{i}$.
\begin{proof}[Proof of Theorem \ref{T:contour3}]
Note that if $E_{0}$ is an eigenvalue of $H_{L}$ for $L$ large i.e., $E_{0}= \ld^{i}_{0}$, we have $\mathcal R^{i} = \emptyset$. Let's assume now that $E_{0}$ is not an eigenvalue of $H_{L}$ for $L$ large.\\
First of all, we will check that the rescaled resonance equation \eqref{eq:4.0} in $\Omega^i$ can be replaced by $f_L(z)=-\frac{e^{-i \theta(E_0)}}{L}$.\\ Indeed, along the segment $A_3B_3$, $f_L(z)$ is real. Along $B_3C_3$, $|f_L(z)|$ is big. Along $C_3D_3$, $|\text{Im}f_L(z)|$ is big. 
Hence, to prove Lemma \ref{L:compare} for $\Omega^{i}$, it suffices to check that 
$$|f_L(z)| \gtrsim \frac{1}{\ve L} \text{ on $A_3D_3$}.$$ 
Put $z=iy \in A_3D_3$ with $0\geq y \geq -\frac{1}{\ve L}$. Assume that $\ld^{i}_{k} \geq E_{0} + 2\ve_{1}$ with $\ve_{1}\asymp \ve^{2}$ for all $k> \ve L$ and $\ld^{i}_{k} \in B_{i}$. Then, by Lemma \ref{L:outsider1}, we have 
\begin{equation}\label{eq:4.57}
\text{Re}f_L(z)= \sum_{k=0}^{\ve L} \frac{\tia^{i}_k \til^{i}_k}{(\til^i_k)^{2}+y^2} + O\left( \frac{1}{\ve_{1}L} \right). 
\end{equation}
For any $\lambda_k \notin B_{i}$, $|\til_k|=L^2|\lambda_k-E_0| \gtrsim L^2$. On the other hand, if $\lambda_k \in B_i$, we have $\lambda_{k} \ne E_{0}$ and $|\til_k| =L^2|\lambda_k-E_0| \geq L^2|\ld^i_0-E_0| \gtrsim 1$. Hence, $|\til_k| \geq \frac{1}{C} \gg |y|$ for all $\lambda_k$. On the other hand, $\til^{i}_{k}>0$ for all $\ld^{i}_{k} \in B_{i}$. Consequently, 
\begin{equation}\label{eq:4.61}
\sum_{k \leq \ve L} \frac{\tia^{i}_k \til^{i}_k}{(\til^{i}_k)^{2}+y^2} \asymp \sum_{k \leq \ve L} \frac{\tia^{i}_{k}}{\til^{i}_{k}} \asymp  \sum_{k=1}^{\ve L}\frac{1}{k^{2}} \asymp 1.
\end{equation}
The estimates \eqref{eq:4.57} and \eqref{eq:4.61} yield $\text{Re}f_L(z) \asymp 1$ on $A_{3}D_{3}$. Hence, Lemma \ref{L:compare} holds true for $\Omega^{i}$.\\
Next, we will study the image of the contour $A_{3}B_{3}C_{3}D_{3}$ under $f_{L}$.\\
On $A_3B_3$, $f_L(z)$ is real and strictly increasing. Hence
\begin{align}\label{eq:4.62}
\sum_{k=0}^L \frac{\tia_k}{\til_k} \leq f_L(z) \leq \sum_{k=0}^L \frac{\tia_k}{\til_k-\til^{i}_0+\delta_1}
\end{align}
where $C$ is a positive constant.\\
Thanks to Lemma \ref{L:outsider1} and \eqref{eq:4.61}, it is easy to see that $\sum\limits_{k=0}^L \frac{\tia_k}{\til_k} \asymp 1$. 
Similarly, we have
\begin{equation}\label{eq:4.67}
\sum_{k=0}^L \frac{\tia_k}{\til_k-\til^{i}_0+\delta_1}= \frac{\tia^{i}_{0}}{\delta_{1}} + \sum_{k=1}^{\ve L} \frac{\tia^{i}_k}{\til^{i}_k-\til^{i}_0+\delta_1} + O\left( \frac{1}{\ve_{1} L}\right) \asymp \frac{1}{\delta_{1}}. 
\end{equation}
Hence, $f_{L}(z) \asymp 1$ on the interval $A_{3}B_{3}$.\\
Next, we consider the segment $A_3D_3$. Since $|\til_{k}| \gtrsim 1 \gg |y|$ for all $\ld_{k} \in \Sigma_{\zz}$, we have
\begin{align}\label{eq:4.63}
\frac{\partial}{\partial y}\text{Im}f_L(z) = \text{Re}[f'_L(z)]=\sum_{k=0}^L \frac{\tia_k}{\til_k^2}\cdot
\frac{1-\frac{y^2}{\til_k^2}}{ \left[ 1+\frac{y^2}{\til_k^2}  \right]^2}\gtrsim \sum_{k=0}^L \frac{\tia_k}{\til_k^2}
\end{align}
where $z=iy$ with $0\geq y \geq -\frac{1}{\ve L}$.\\ 
We will show that, for all $z= x+iy \in \Omega^{i}$,
\begin{equation}\label{eq:4.68}
\sum_{k=0}^L \frac{\tia_k}{(\til_k-x)^2+y^2} \asymp 1.
\end{equation}
Indeed, since $|\til_k-x| \gg |y|$ for all $z=x+iy\in \Omega^{i}$ and all $\ld_{k} \in \Sigma_{\zz}$, we have $(\til_k-x)^2+y^2 \asymp (\til_k-x)^2$. Then, argument as in \eqref{eq:4.61}, \eqref{eq:4.67}, we have \eqref{eq:4.68} follow.\\ 
Consequently, $\text{Im}f_L(z)$ is strictly increasing on $A_{3}D_{3}$ and $|f'_L(z)| \gtrsim 1$ on $A_{3}D_{3}$.\\
Now, we give estimates on the real and imaginary parts of $f_L(z)$ on $A_3D_3$.
\begin{align}\label{eq:4.64}
0\geq \text{Im}f_L(z) \geq \text{Im}f_L(D_3) = \text{Im}f_L\left (-\frac{i}{\ve L}\right) \asymp -\frac{1}{\ve L} \left( \sum_{k=0}^L \frac{\tia_k}{\til_k^2} \right) \asymp -\frac{1}{\ve L}.
\end{align}
Besides, as we proved before, $\text{Re}f_{L}(z) \asymp 1$ on $A_{3}D_{3}$.\\
Similarly, we have the same conclusion for $f_{L}(z)$ on $B_{3}C_{3}$.\\
Finally, we study $f_L$ on $C_3D_3$. Let $z\in C_3D_3$, $z=x+i y$ where $x\in [0, \til^{i}_0-\delta_1]$ and $y\equiv  -\frac{1}{\ve L}$.\\
Using the equation \eqref{eq:4.13}, we can check easily that $\text{Re}[f'_L(z)]$ is bigger than a positive constant on $C_3D_3$.
Hence, $\text{Re}f_L(z)$ is strictly increasing in $x$. 
Consequently, on $C_3D_3$,
\begin{equation}\label{eq:4.65}
 1 \asymp \text{Re} f_L(D_3) \leq \text{Re} f_L(z) \leq \text{Re} f_L(C_3) \asymp 1. 
\end{equation}
Finally, we compute the magnitude of $\text{Im}f_L(z)$. For $z=x+iy \in C_3D_3$, \eqref{eq:4.68} yield
\begin{equation}\label{eq:4.66}
\text{Im}f_L(z) = y \left( \sum_{k=0}^L \frac{\tia_k}{(\til_k-x)^2+y^2}  \right) \asymp -\frac{1}{\ve L}.
\end{equation}
To sum up, $f_L$ is bijective from $\Omega^i$ to $f_L(\Omega^i)$ and $|f'_L(z)| \gtrsim 1$ for all $z\in \Omega^i$. Moreover, there exists a positive constant $c$ such that $\text{dist}(0, f_L \left(\Omega^i)\right) \geq c$. Hence, $-\frac{e^{-i \theta(E_0)}}{L} \notin f_L(\Omega^i)$ which implies that there are no resonances in $\Omega^i$.
\end{proof}

\def\cprime{$'$}

\noindent {\tiny (Trinh Tuan Phong) Laboratoire Analyse, G\'{e}om\'{e}trie $\&$ Applications,\\ 
	UMR 7539, Institut Galil\'{e}e, Universit\'{e} Paris 13, Sorbonne Paris Cit\'{e},\\ 
	99 avenue J.B. Cl\'{e}ment, 93430 Villetaneuse, France\\
\emph{Email}: trinh@math.univ-paris13.fr 
}

\end{document}